\documentclass[12pt, draftclsnofoot, onecolumn]{IEEEtran}
\usepackage[colorlinks,
linkcolor=red,
anchorcolor=green,
citecolor=blue
]{hyperref} 
\usepackage{amsfonts,amssymb,amsmath,array}
\usepackage{latexsym}
\usepackage{CJK}
\usepackage{cite}
\usepackage{bm}
\usepackage{soul}

\usepackage{graphicx}
\usepackage{epstopdf}
\usepackage{epsfig}
\usepackage{subfigure} 

\usepackage{verbatim}  

\usepackage[usenames,dvipsnames]{color}
\definecolor{mycyan}{cmyk}{.3,0,0,0}
\usepackage{colortbl}
\usepackage{enumerate}
\usepackage{cases}

\usepackage{mathrsfs}	
\usepackage{algorithm} 
\usepackage{algorithmic} 
\usepackage{booktabs}
\usepackage{textcomp}
\usepackage{multirow}
\usepackage{lettrine}
\usepackage[scaled=0.92]{helvet}

\DeclareMathOperator*{\argmin}{argmin}

\usepackage{algorithm}
\usepackage{algorithmic}

\usepackage{amsthm} 

\DeclareMathOperator*{\argmine}{arg\, }

\newtheorem{thm}{Theorem}
\newtheorem{lem}{Lemma}
\newtheorem{prop}{Proposition}
\newtheorem{rem}{Remark}

\begin{document}
\title{Beyond Empirical Models: Pattern Formation Driven Placement of UAV Base Stations}

\author{Jiaxun~Lu, Shuo~Wan, Xuhong~Chen, Zhengchuan~Chen,~\IEEEmembership{Member, ~IEEE}, Pingyi~Fan,~\IEEEmembership{Senior Member, ~IEEE}, Khaled~B.~Letaief,~\IEEEmembership{Fellow, ~IEEE}\\

\thanks{
Jiaxun~Lu, Shuo~Wan, Xuhong~Chen and Pingyi Fan are with Tsinghua National Laboratory for Information Science and Technology(TNList) and the Department of Electronic Engineering, Tsinghua University, Beijing, P.R. China, 100084. E-mail: \{lujx14,~s-wan13,~chenxh13\}@mails.tsinghua.edu.cn, fpy@tsinghua.edu.cn.}
\thanks{Zhengchuan Chen is with the College of Communication Engineering, Chongqing University, Chongqing, P.R. China, 400044. E-mail: chenzc.ee@gmail.com.} 
\thanks{Khaled B. Letaief is with the Department of Electrical and Computer Engineering, HKUST, Hong Kong (email: eekhaled@ust.hk) and Hamad bin Khalifa University, Qatar (email: kletaief@hbku.edu.qa).}
}

\maketitle

\graphicspath{{Figures/}}

\begin{abstract}

This work considers the placement of unmanned aerial vehicle base stations (UAV-BSs) with criterion of minimum UAV-recall-frequency (UAV-RF), indicating the energy efficiency of mobile UAVs networks. Several different power consumptions, including signal transmit power, on-board circuit power and the power for UAVs mobility, and the ground user density are taken into account. Instead of conventional empirical stochastic models, this paper utilizes a pattern formation system to track the instable and non-ergodic time-varying nature of user density. We show that for a single time-slot, the optimal placement is achieved when the transmit power of UAV-BSs equals their on-board circuit power. Then, for multiple time-slot duration, we prove that the optimal placement updating problem is a nonlinear dynamic programming coupled with an integer linear programming. Since the original problem is NP-hard and can not be solved with conventional recursive methods, we propose a sequential-Markov-greedy-decision method to achieve near minimal UAV-RF in polynomial time. Further, we prove that the increment of UAV-RF caused by inaccurate predicted user density is proportional to the generalization error of learned patterns. Here, in regions with large area, high-rise buildings or low user density, large sample sets are required for effective pattern formation.

\end{abstract}


\IEEEpeerreviewmaketitle

\section{Introduction}\label{Sec:Introduction}

Unmanned aerial vehicle base stations (UAV-BSs) have been considered  as a promising solution to provide wireless coverage in a rapid manner. In this system, UAVs are often powered by batteries \cite{gupta2016survey}, which limits its life-time. Here, UAV-BSs are usually recalled periodically and there is an urgent need to increase the energy-efficiency of UAVs systems. A general strategy to improve energy-efficiency is adjusting the placement of UAVs according to ground user density, working environment and desired transmit data rate etc.\cite{Mozaffari2016Efficient,kalantari2016number,mozaffari2016optimal}.

In the literatures, the placement of UAV-BSs has been studied under a broad range of aspects\cite{al2014optimal,bor2016efficient,lyu2017placement,Mozaffari2016Efficient,kalantari2016number,mozaffari2016optimal,alzenad20173d}. The optimal hovering altitude that maximizes coverage radius was discussed in \cite{al2014optimal}. Later, the authors in \cite{Mozaffari2016Efficient,alzenad20173d} investigated the optimal hovering altitude and coverage radius. In \cite{bor2016efficient}, the authors analyzed the optimal placement of UAV-BSs simultaneously maximizing the number of covered users and energy-efficiency. In \cite{kalantari2016number},\cite{mozaffari2016optimal} and \cite{lyu2017placement}, they discussed the relationships among the optimal placement, minimum required number of UAV-BSs, as well as the density of ground users. The main focus of these works is on minimizing transmit power. In fact, the on-board circuit power consumption, related to rotors, computational chips and gyroscopes etc., and the potential mobility power consumption of UAVs may also affect the life-time of network and should be taken into account in the viewpoint of system. Up to now, the works on energy-efficient placement of UAV-BSs considering on-board circuit power and mobility power are quite limited. A power model investigating the peculiar features of UAVs, like available energy, weight, maximum speed, etc. was formulated in \cite{di2015energy}. The on-board circuit power and mobility power of UAVs were addressed in \cite{lu2017energy} and \cite{mozaffari2017mobile}, separately. These works provide comprehensive analyses based on a time-invariant density of ground users. However, in practice, due to the directional shift of people consequences for transportation design or epidemic control etc., the human flow usually follows Turing pattern \cite{asllani2014theory}, which leads to time-varying, instable and non-ergodic density of ground users, raising some new challenges to the placement of UAV-BSs.

In conventional stable and ergodic cases, one can develop an empirical stochastic model for density of ground users (e.g., Markov modulated) where the statistical parameters can be estimated from data \cite{sharma2010optimal}. Such models may not be applied to the instable scenarios. Fortunately, the developing data-driven methods cast a bright light on challenges in wireless communication. For example, useful information extraction with Data-mining approach was utilized in \cite{bacstuug2015big} to enhance the caching performance of wireless networks. Machine-learning techniques have been applied to decision-making and feature classification in cognitive radio problems (see \cite{bkassiny2013survey} and the references therein). Some tutorial works on applying data-driven methods to the wireless domain were presented in \cite{alsheikh2014machine} and \cite{kulin2016data}. These works deliver some useful motivations to us. Recently, it was shown in \cite{dias2016survey} and \cite{lee2017deep} that using machine learning, it is possible to construct a more intelligent context-aware pattern formation system by predicting future situations as well as monitoring the current state.  

\begin{figure}[htbp]
\centering
\includegraphics[width=0.6\textwidth]{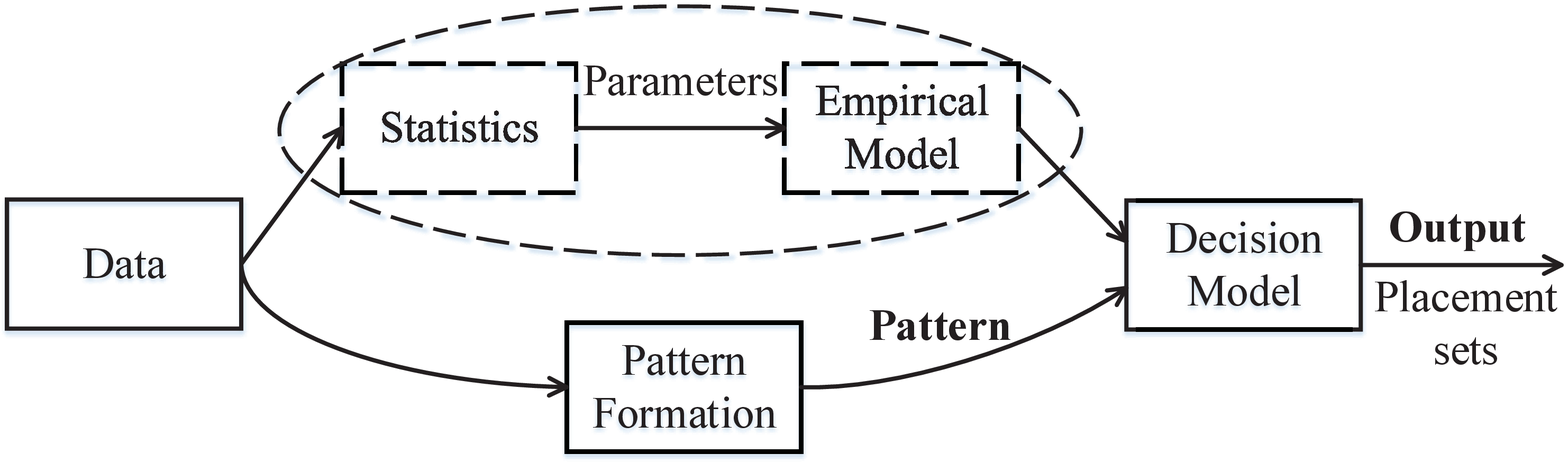}
\caption{Pattern Formation aided framework of our considered system.} \label{Fig:MLModel}{}
\end{figure}

As shown in Fig. \ref{Fig:MLModel}, we propose a framework, using a pattern formation module rather than a statistical empirical model, to track the instable and non-ergodic time-varying nature of ground user densities. Then, we address the importance of accurate pattern formation and consider the decision model on optimal placement that maximizes the life-time of the mobile UAVs network. In this case, the on-board circuit power and the potential mobility power of UAVs are also considered. This paper focuses on the downlink of UAV-BSs, in which each of the ground users is served with fixed data rate. In the system, the considered duration is partitioned into continuous time-slots. At the beginning of each time-slot, UAV-BSs are allowed to decide whether or not to update their placement according to ground user density. To our best knowledge, this paper is \emph{one of the first comprehensive studies on the joint optimal deployment of UAV-BSs and pattern formation in scenario with instable and non-ergodic time-varying density of ground users}. 

To characterize the life-time of mobile UAVs network, we employ the notion of  UAV-recall-frequency (UAV-RF), the frequency of the active UAVs run out of batteries, as the physical index. That is, maximizing life-time is equivalent to minimizing UAV-RF. In this direction, we first consider the optimal placement of UAV-BSs that minimizes UAV-RF in one time-slot and then extend the discussion to multi-slot duration. In fact, the UAV-RF in one time-slot can be treated as static UAV-RF. In this case, by analyzing the coverage scenario with one single UAV, we prove that the optimal hovering altitude minimizing transmit power is proportional to the coverage radius, and the slope is only determined by communication environment (high-rise urban, dense urban and urban, etc.), which is a general extension of previous results in \cite{al2014optimal} by considering the density of users inside the coverage of UAV-BSs. More specifically, in environment with high-rise buildings, the slope is large, and hence UAVs are supposed to fly higher compared with environment with low-rise buildings.

By applying the derived optimal hovering altitude, we investigate the static UAV-RF versus environment, coverage parameters and on-board circuit power, where coverage parameters include the coverage radius, user density and desired data rate. Analytical results demonstrate that: 1) The minimal static UAV-RF is achieved when transmit power equals the on-board circuit power; 2) The minimal static UAV-RF becomes large in scenarios with high-rise buildings, high on-board circuit power, and large user density and data rate. This indicates that limiting on-board circuit power can effectively prolong the life-time of mobile UAVs network. Compared with the optimal coverage radius given in \cite{alzenad20173d}, our results provide more insights on the design of UAVs networks by investigating the on-board circuit power.

For the multiple time-slot case, it requires to decide the optimal placement updating epochs of UAV-BSs in cases with instable time-varying density of ground users. The corresponding UAV-RF is denoted as dynamic UAV-RF, which is relevant to transmit power, on-board circuit power and potential mobility power of UAV-BSs. We show that the placement optimization problem is a multi-stage decision process and can be written as a nonlinear dynamic programming (NLDP) coupled with an inherent integer linear programming (ILP). For the inherent ILP, we propose a polynomial time solution by transforming it into a standard assignment problem with some amendments. However, for the NLDP, we show that it can not be solved with conventional methods, such as recursive manners, because the number of update epochs is unknown in advance.

Noticing the NP-hardness of NLDP, we shall design a sequential-Markov-greedy-decision (S-MGD) method to find the near-optimal solution by utilizing the notion of Pareto-optimality, which is proved to be with polynomial complexity. Simulation results show that our proposed S-MGD method can stably achieve near-optimal performance in terms of minimum dynamic UAV-RF. In particular, when mobility power is extremely low compared with on-board circuit power and transmit power, the proposed S-MGD method updates the placement of UAV-BSs at the beginning of each time-slots. By contrast, in cases with extremely high mobility power, UAV-BSs hold their placement during the considered duration. 

Finally, the relationships among sampling number, density pattern accuracy and increment of UAV-RF are characterized in detail. These results imply that in subregions with large area, high-rise buildings and low user density, large sample sets are needed for effective pattern formation and reducing UAV-RF. Specifically, we first prove that the increment of  UAV-RF caused by inaccurate density patterns is proportional to the generalization error. Then, we theoretically derive the minimum sampling number of each subregion with the Vapnik-Chervonenkis (VC) theorem\cite{vapnik2013nature}, where the overall increased UAV-RF is upper-bounded.

The rest of this paper is organized as follows. In Section \ref{Sec:SysModel}, the system model is introduced, and an available density pattern is presented as the test set in this paper. Then, the optimal placement of UAV-BSs minimizing static UAV-RF is discussed in Section \ref{Sec:StaticCoverage}. In Section \ref{Sec:DynamicCoverage}, the S-MGD based placement updating method is presented to minimize dynamic UAV-RF. Section \ref{sec:insight} analyzes the effects of the accuracy of pattern formation system. In Section \ref{Sec:Simulations}, the validity of previous theoretical results and the effectiveness of our proposed methods are verified by numerical results. Finally, conclusions are given in Section \ref{Sec:Conclusion}.

\section{Coverage Model, Air-to-Ground Channel and User Density Pattern}\label{Sec:SysModel}

In this section, we shall first illustrate the downlink coverage model of UAV-BSs, where instable time-varying user density is considered. Then, we introduce the air-to-ground (A2G) channel and show the existence of optimal hovering altitude of UAV-BSs. Finally, an available pattern of previous time-varying user density is presented as the test set in this paper.

\subsection{UAV-BS Coverage Model}\label{SubSec:UAVCoverageModel}

As shown in Fig. \ref{Fig:UAVCoverage}, a geographical area is divided into several subregions according to the municipal planning of modern cities and their different ground user density patterns\cite{xu2016understanding}. Consider a length $T$ duration, let the time-varying density of ground users be $\lambda_{\beta}\left(t\right)$, where $\beta\in\{1,2,\cdots,\kappa\}$ and $t\in[0,T]$ index subregion and time, respectively. $\kappa$ is the number of subregions. For the notation simplicity, unless otherwise stated, we drop the time index from the equations. To balance the number of active UAV-BSs between adjacent time-slots, a recall and supplement center (RSC) is deployed in considered area. 

\begin{figure}[htbp]
\centering
\includegraphics[width=0.48\textwidth]{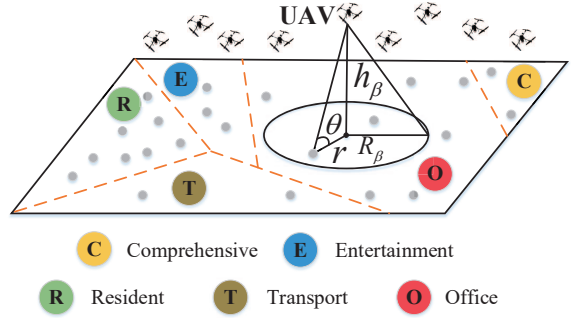}
\caption{A typical area classified into different subregions according to user density patterns and the UAV coverage model therein.} \label{Fig:UAVCoverage}
\end{figure}

We focus on the downlink of network in which UAV-BSs adopt a frequency division multiple access technique to serve each of the ground users with fixed data rate $C$. UAV-BSs assign individual frequency bands to mobile ground users, and hence the frequency interference between UAV-BSs is avoided. We assume that the transmit power of each UAV-BS and the available bandwidth are sufficient to meet the rate requirement of users. In this case, we consider a disk-covering model where UAV-BSs in the same subregion provide equal coverage radius\cite{mozaffari2016unmanned,lyu2017placement}. Here, considering the fact that the area of overlaps between adjacent disks are proportional to the area of disks, without loss of generality, we express the required number of UAV-BSs in the $\beta$-th subregion at $t$ as{\footnote{Minimizing required number of UAV-BSs can be formulated as the geometric disk cover problem\cite{srinivas2009construction}, and can be optimally solved by core-sets method\cite{fayed2013mixed}.}}
\begin{equation}\label{Equ:UAVNumber}
N_\beta(t) = \frac{S_\beta}{\pi R^2_{\beta}(t)},
\end{equation}
where $S_\beta$ is the area of the $\beta$-th subregion and $R_{\beta}(t)$ is the corresponding coverage radius.


\subsection{Air-to-Ground Channel}\label{SubSec:A2GChannel}

The A2G channel can be characterized into line-of-sight (LOS) link or non-line-of-sight (NLOS) link\cite{al2014optimal}, and the path loss therein can be given by
\begin{equation}\label{Equ:A2GPathLoss}
L_\xi(r,h_\beta) = \begin{cases}
\left( 4\pi f/c \right)^2 \left( r^2+h_\beta^2 \right)\,\eta_{0},& {\xi=0}\\
\left( 4\pi f/c \right)^2 \left( r^2+h_\beta^2 \right) \,\eta_{1},& {\xi=1},
\end{cases}
\end{equation}
where $\xi=0$ and $\xi=1$ denote LOS link and NLOS link, respectively. $f$ is the carrier frequency and $c$ is the traveling speed of light. $r\in[0,R_{\beta}]$ is the distance between the user of interest and the projection of UAV-BS on ground, and $h_\beta$ is the hovering altitude of UAV-BSs. $\eta_{0}$ and $\eta_{1}$ are the excessive path loss on the top of the free space path loss (FSPL) for LOS and NLOS links, determined by communication environment (suburban, urban, dense urban, high-rise urban or others). Typically, $\eta_{1}\gg\eta_{0}$ as the obstacles in propagation paths greatly enhance the path loss of NLOS link.

The average path loss of A2G channel is also determined by the LOS probability \cite{al2014optimal}
\begin{equation}\label{equ:ProbofLOSLink}
P_{0}(r, h_\beta) = \frac{1}{1 + a {\rm{exp}}(-b[\theta-a])},
\end{equation}
where $a$ and $b$ are constants determined by environment, and $\theta=\frac{180}{\pi} \tan^{-1}(h_\beta/r)$ is the elevation angle shown in Fig. \ref{Fig:UAVCoverage}. Then, $1-P_{0}(r, h)$ denotes the NLOS probability. The average path loss of A2G channel can be derived as
\begin{equation}\label{Equ:AveragePathLoss}
\begin{split}
\bar{L}(r, h_\beta) &= P_{0}(r, h_\beta)L_0(r, h_\beta) + \left(1-P_{0}(r, h_\beta)\right)L_1(r, h_\beta)\\
&=\underbrace{(4\pi f/c)^2 \left( r^2+h_\beta^2 \right)}_{\rm{FSPL}} \, \underbrace{\big( \eta_1 + P_{0}(r, h_\beta)(\eta_0-\eta_1) \big)}_{\rm{average\,excessive\,path\,loss}}.
\end{split}
\end{equation}
This clearly characterizes the individual effects of FSPL and average excessive path loss. The first part accounts for FSPL, which monotonically increases with $h_\beta$ due to the growing distance between UAV and user; the second part, which represents the average excessive path loss, is monotonically decreasing with $h_\beta$. This is because large $h_\beta$ leads to high LOS probability of A2G channel. For a specific coverage radius $R_{\beta}$, \eqref{Equ:AveragePathLoss} implies the existence of optimal hovering altitude minimizing average path loss of A2G channel. Besides, to minimize the UAV-RF of considered area, the optimal coverage radius also needs to be jointly considered.


\subsection{An Available Density Pattern as the Test Set}\label{sec:learnedDensityPattern}

The empirical average traffic amount of the $\beta$-th subregion can be reconstructed as $\{x_{\beta}^{\rm{r}}[n]\}$ by the inverse discrete Fourier transformation (IDFT)\cite{xu2016understanding}:
\begin{equation}\label{Equ:ReconstructTraffic}
x_{\beta}^{\rm{r}}[n] = \frac{\gamma_{\rm{r},\beta}}{N} \sum_{k=0}^{N-1} {X}_{\beta}^{\rm{r}}[k]\exp(2\pi j k n/N),
\end{equation}
where $n\in[0,N]$ is the sampling index, and with the sampling period $\mu=10$ min, the sampling number in 4 weeks is $N=4032$. $\gamma_{\rm{r},\beta}$ is the reconstruction scaling factor at the $\beta$-th subregion, and $X_{\beta}^{\rm{r}}[k]$ is the IDFT coefficient expressed by
\begin{equation}
{X}_{\beta}^{\rm{r}}[k] = 
\begin{cases}
{X}_{\beta}[k], & k = 0,4,28,56,N-4,N-28,N-56\\
0, & \text{otherwise}. 
\end{cases} 
\end{equation}
${X}_{\beta}[k]$ is the frequency-domain coefficients of DFT($x_{\beta}[n]$), where $x_{\beta}[n]$ is the sampled time-domain traffic amount in the $\beta$-th subregion. According to the properties of IDFT, ${X}_{\beta}[k] = {X}^{\dagger}_{\beta}[N-k]$ for the real time-domain traffic amount, where $(\cdot)^{\dagger}$ denotes the conjugate transpose of $(\cdot)$. Hence, only parts of the coefficients and the scaling factor $\gamma_{\rm{r},\beta}$ are listed in Table \ref{tab:ReconstructCoefficients}.

\begin{table}[htbp]
\small
 \caption{Coefficients of Reconstruction Function}\label{tab:ReconstructCoefficients}
 \centering
 \begin{tabular}{cccccc}
  \toprule
  Subregions& $\gamma_{\rm{r}}$ $(\times10^{11})$ & $k=0$ $(\times10^{-4})$ & $k=4$  & $k=28$  &  $k=56$\\
  \midrule
  \rowcolor{mycyan}\texttt{E} & $8.35$ &3.24 & $0.06e^{-0.3j}$ & $0.5e^{2.36j}$ & $0.08e^{0.69j}$\\
  
  \texttt{R} & $17.4$ &2.73 & $0.04e^{-1.02j}$ & $0.27e^{1.72j}$ & $0.17e^{1.35j}$\\
  
  \rowcolor{mycyan}\texttt{T} & $4.32$ &3.73 & $0.1e^{1.04j}$ & $0.38e^{2.53j}$ & $0.28e^{2.46j}$\\
  
  \texttt{O} & $5.23$ &4.63 & $0.21e^{1.21j}$ & $0.56e^{2.52j}$ & $0.2e^{0.29j}$\\
  
  \rowcolor{mycyan}\texttt{C} & $17.4$ &2.85 & $0.04e^{0.35j}$ & $0.3e^{2.19j}$ & $0.15e^{1.11j}$\\
  \bottomrule
 \end{tabular}
\end{table}

Then, with the assumption that each of the ground users is served with fixed data rate $C$, the average user density in the $\beta$-th subregion at $t$ can be immediately expressed as
\begin{equation}\label{equ:ConvertBetweenLambdaAndTrafficAmount}
\lambda_{\beta}(t) =  \frac{1}{C S_{\rm{bs}}} \, x_{\beta}^{\rm{r}}\left[ \left\lfloor {t}/{\mu} \right \rfloor \right].
\end{equation}
$S_{\rm{bs}}$ is the area of the coverage of base stations shown in \cite{xu2016understanding}. The normalized density patterns in one week corresponding to five considered subregions are depicted in Fig. \ref{Fig:TrafficPatterns}. It can be observed that the density patterns are instable and show the breath-out mode in the considered duration. In this case, to provide optimal coverage in each subregion, the placement of UAV-BSs is supposed to be updated with respect to the learned time-varying ground user densities. However, updating the placement requires the mobility of UAV-BSs, which may cost much energy and reduce the life-time of mobile UAVs network \cite{di2015energy}. Therefore, to minimize the UAV-RF of network, the optimal placement of UAV-BSs in one time-slot and the placement updating strategy in the considered duration should be jointly optimized.

\begin{figure}[tbp]
\centering
\includegraphics[width=0.48\textwidth]{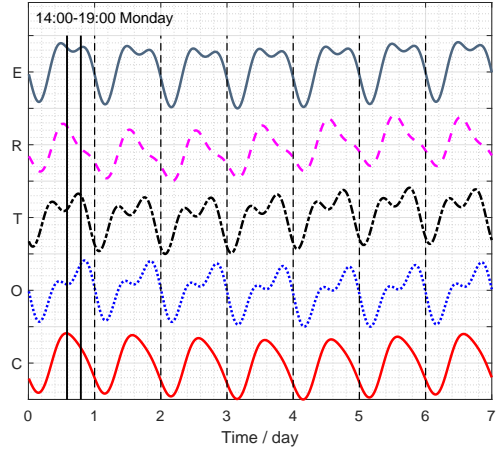}
\caption{Density patterns of ground users in different subregions.} \label{Fig:TrafficPatterns}
\end{figure}

\section{Optimal Placement of UAV-BSs in single time-slot}\label{Sec:StaticCoverage}

In this section, we first formulate the optimal UAV-BSs placement problem in single time-slot. Then, we show that the original problem can be solved by separately investigating the optimal hovering altitude and optimal coverage radius. Finally, we present the optimal condition with respect to on-board circuit power. 

\subsection{Problem Formulation}\label{SubSec:ProblemFormulation}

Recall that the downlink rate is fixed to $C$. Denote the allocated transmit power to the interested user located at $r$ be $P_{\rm{tr},\rm{user},\xi}\big(r,h_\beta\big)$. Then, following Shannon formula, one has
\begin{equation}
C = W\log_2\left( 1 + \frac{P_{\rm{tr},\rm{user},\xi}\big(r,h_\beta\big)}{L_\xi\big(r,h_\beta\big) N_0W} \right),
\end{equation}
where $N_0$ is the noise power spectrum density, and $W$ is the allocated bandwidth to interested user. The relevant transmit power can be immediately derived as
\begin{equation}
P_{\rm{tr},\rm{user},\xi}\big(r,h_\beta\big) = L_\xi(r,h_\beta) N_0W \left( 2^{C/W} - 1\right).
\end{equation}
Accordingly, the average transmit power is given by
\begin{equation}
\bar{P}_{\rm{tr},\rm{user}}\big(r,h_\beta) = \bar{L}(r, h_\beta\big) N_0W \left( 2^{C/W} - 1\right).
\end{equation}
The expected value of transmit power of UAV-BS is the integral of average transmit power relevant to all users inside the coverage. That is,
\begin{equation}\label{equ:CommunicationPower}
\begin{split}
P_{\rm{tr}}\big(R_{\beta}, \lambda_{\beta}, h_\beta\big) = \lambda_{\beta}\int_0^{R_{\beta}} 2\pi r \bar{P}_{\rm{tr},\rm{user}}\big(r, h_\beta\big) \,dr.
\end{split}
\end{equation}

Let the on-board circuit power and the battery capacity of one single UAV-BS be $P_{\rm{cu}}$ and $E_{\rm{b}}$, respectively. Then, the average life-time of UAV-BSs in the $\beta$-th subregion is 
\begin{equation}
{E_{\rm{b}}}/{\big(P_{\rm{tr}}(R_{\beta}, \lambda_{\beta}, h_\beta) + P_{\rm{cu}}\big)},
\end{equation}
and accordingly the static UAV-RF at the $\beta$-th subregion can be expressed as 
\begin{equation}\label{Equ:ChangingBatteryNum}
\Phi_{\rm{st},\beta}(t) = \frac{N_\beta(t)\left(P_{\rm{tr}}(R_{\beta}, \lambda_{\beta}, h_\beta) + P_{\rm{cu}}\right)}{E_{\rm{b}}}.
\end{equation}
From  \eqref{Equ:ChangingBatteryNum}, it can be observed that static UAV-RF is determined by the number of UAV-BSs and the power relevant to signal transmission and on-board circuit. Besides, the total consumed power of  the mobile UAVs network can be written as 
\begin{equation}
\sum_{\beta=1}^{\kappa}N_\beta(t)\bigg(P_{\rm{tr}}\big(R_{\beta}, \lambda_{\beta}, h_\beta\big) + P_{\rm{cu}}\bigg)  = E_{\rm{b}} \sum_{\beta=1}^{\kappa} \Phi_{\rm{st},\beta}(t).
\end{equation}
That is, the total UAV-RF of our considered area takes both the power consumed by one single UAV-BS and the number of UAV-BSs into account. Compared with the total consumed power, the notion of UAV-RF characterizes the frequency of active UAVs run out of batteries and is a more comprehensive indicator of the life-time of mobile UAVs network.

As illustrated in Section \ref{SubSec:A2GChannel}, to minimize static UAV-RF of considered area in single time-slot, the altitudes and coverage radii of UAV-BSs need to be jointly considered. That is,
\begin{flalign}\label{Equ:MinChangingBatteryNumber}
\mathcal{P}_1: \,\,\min_{\left\{R_{\beta},\,h_\beta\right\}} \,\,\, &\sum_{\beta=1}^{\kappa} \Phi_{\rm{st},\beta}(t)\\
\textrm{s.t.}\,\,\,&h_\beta\geq0,R_{\beta} > 0.\label{equ:MainOptimizeProblemConstraints}
\end{flalign}
In $\mathcal{P}_1$, \eqref{equ:MainOptimizeProblemConstraints} corresponds to the non-negativity of hovering altitudes and coverage radii of UAV-BSs. According to \eqref{Equ:ChangingBatteryNum}, it can be observed that the optimization variables $\left\{R_{\beta},h_\beta\right\}$ are coupled, and hence directly solving problem $\mathcal{P}_1$ is difficult. In the following, we consider a two-stage method instead. 

\subsection{Optimal Hovering Altitude of UAV-BSs}\label{Sec:Coverage4SingleUAV}
We firstly consider the optimal hovering altitude minimizing the transmit power with a fixed coverage radius, which can be expressed as
\begin{flalign}
\mathcal{P}_{1\text{-}\rm{A}}: \,\,\min_{\left\{h_\beta\right\}} \,\,\, &P_{\rm{tr}}\big(R_{\beta}, \lambda_{\beta}, h_\beta\big)\\
\textrm{s.t.}\,\,\,&h_\beta\geq0.\label{equ:OptimalAltitudeConstraints}
\end{flalign}
\eqref{equ:OptimalAltitudeConstraints} reflects the non-negativity of hovering altitude. Then, the solution to $\mathcal{P}_{1\text{-}\rm{A}}$ can be summarized as follows.
\begin{lem}\label{Lem:ScaleOptimalAltitude}
1) The transmit power of single UAV-BS can be expressed as
\begin{equation}\label{Equ:ScaleTransmitPower}
\begin{split}
P_{\rm{tr}}\big(R_{\beta}, \lambda_{\beta}, h_\beta\big) =\gamma_{\rm{tr}} P_{\rm{tr},1}\big({h_\beta}/{R_{\beta}}\big),
\end{split}
\end{equation}
$P_{\rm{tr},1}(h_\beta/R_{\beta})=P_{\rm{tr}}\big(1, 1, h_\beta/R_{\beta}\big)$ is the transmit power corresponding to hovering altitude $h_\beta/R_{\beta}$, when $R_{\beta}$, $\lambda_\beta$ and $C$ are normalized. $\gamma_{\rm{tr}}=\lambda_{\beta} R^4_{\beta} \left( 2^{C/W} - 1\right)$ is the relevant scaling factor.

2) The optimal hovering altitude corresponding to specific $R_{\beta}$ is $h^*=R_{\beta}h_{\beta,1}^*$,
where $h_{\beta,1}^*$ is the optimal hovering altitude that minimizes $P_{\rm{tr},1}(h_{\beta,1})$ and is only determined by communication environment.
\end{lem}
\begin{proof}
With \eqref{equ:CommunicationPower}, we can derive that
\begin{equation}\label{Equ:KernelFunction}
P_{\rm{tr},1}\left({h_\beta}/{R_{\beta}}\right) = \int_{0}^{1} 2\pi r \bar{L}(r,h_\beta/R_{\beta})N_0 W {\rm{d}} r.
\end{equation}
Besides,  \eqref{Equ:AveragePathLoss} can be rewritten as $\bar{L}(r, h_\beta) = R^2_{\beta}\,\bar{L}(r/R_{\beta},h_\beta/R_{\beta})$.
The transmit power of one single UAV-BS can be expressed as
\begin{equation}
\begin{split}
P_{\rm{tr}}(R_{\beta}, \lambda_{\beta}, h_\beta) &= \lambda_{\beta}\int_0^{R_{\beta}} 2\pi r \bar{P}_{\rm{u}}(r, h_\beta) \,{\rm{d}}r\\
&=\lambda_{\beta}R^2_{\beta}\left( 2^{C/W} - 1\right) \int_0^{R_{\beta}} 2\pi r\,\bar{L}(r/R_{\beta},h_\beta/R_{\beta}) {\rm{d}}r\\
&=\lambda_{\beta}(\beta,t)R^4_{\beta}\left( 2^{C/W} - 1\right)\int_0^{1} 2\pi r\,\bar{L}(r,h_\beta/R_{\beta}) {\rm{d}} r\\
&=\lambda_{\beta}R^4_{\beta}\left( 2^{C/W} - 1\right) P_{\rm{tr}}(1, 1, h_\beta/R_{\beta}).
\end{split}
\end{equation}
This completes the proof of \eqref{Equ:ScaleTransmitPower}. Then, one can immediately have 
\begin{equation}
\frac{\partial P_{\rm{tr}}\big(R_{\beta}, \lambda_{\beta}, h_\beta\big)}{\partial h_\beta} = 0
\Leftrightarrow
\frac{\partial P_{\rm{tr},1}(h_\beta/R_{\beta})}{\partial h_\beta} = 0.
\end{equation}
Noticing that $P_{\rm{tr},1}(h_\beta/R_{\beta})$ only accounts for the communication environment, $h_\beta^*=R_{\beta}h_{\beta,1}^*$ naturally follows.
\end{proof}
Lemma \ref{Lem:ScaleOptimalAltitude} clearly indicates the individual effects of communication environment and coverage parameters on optimal hovering altitude. That is, $h_{\beta,1}^*$ accounts for the environmental statistics, while $\gamma_{\rm{tr},1}$ explains the effects of coverage parameters. This implies that solving $\mathcal{P}_{1\text{-}\rm{A}}$ is equivalent to finding the optimal hovering altitude that minimizes $P_{\rm{tr},1}(h_{\beta,1})$. That is,
\begin{equation}\label{Equ:OptimalAltitude}
h_{\beta,1}^* = \argmine_{h_{\beta,1}} \left\{\frac{\partial P_{\rm{tr},1}(h_{\beta,1})}{\partial h_{\beta,1}} = 0\right\}.
\end{equation}
From \eqref{Equ:KernelFunction}, we have that
\begin{equation}\label{Equ:PartialEquation}
\begin{split}
\frac{\partial P_{\rm{tr},1}(h_{\beta,1})}{\partial h_{\beta,1}} = 0 \Leftrightarrow \int_0^1 &\Big\{ 2h_{\beta,1}\big( \eta_1 + P_{0}(r, h_{\beta,1})(\eta_0-\eta_1) \big) + \\
& \big(r^2 + h_{\beta,1}^2\big) \big( \eta_1 + \frac{\partial P_{0}(r, h_{\beta,1})}{\partial h_{\beta,1}}(\eta_0-\eta_1) \big) \Big\} r {\rm{d}} r = 0,
\end{split}
\end{equation}
where based on \eqref{equ:ProbofLOSLink},
\begin{equation}\label{Equ:PartialPLOS}
\frac{\partial P_{0}(r, h_{\beta,1})}{\partial h_{\beta,1}} = \frac{180 b r P_{0}(r, h_{\beta,1})}{\pi (r^2 + h_{\beta,1}^2)}(1-P_{0}(r, h_{\beta,1})).
\end{equation}
Substituting \eqref{Equ:PartialEquation} and \eqref{Equ:PartialPLOS} into  \eqref{Equ:OptimalAltitude}, one can get $h_{\beta,1}^*$. However, it is overwhelming to obtain explicit solution of  \eqref{Equ:OptimalAltitude}. As an alternative, we propose a binary search algorithm to calculate the optimal hovering altitude, as shown in Algorithm \ref{Alg:OptimalAltitude}. Note that the precision $\epsilon=10^{-3}$ and the iteration scaling factor $\gamma_{i}=10$ shown at line 0 of Algorithm \ref{Alg:OptimalAltitude} can be readily replaced with any other values that satisfy $\epsilon > 0$ and $\gamma_{i}>1$.

\begin{algorithm}[htbp]
    \caption{Optimal Hovering Altitude}\label{Alg:OptimalAltitude}
    \begin{algorithmic}[1]
    	\item \textbf{Initialize} ~~\\          
        Environmental parameters: $\eta_1$, $\eta_2$, $a$ and $b$;\\
        Input coverage parameters: $R_{\beta}$, $\lambda_{\beta}$, $C$;\\
        Initialize iteration parameters: $h_{\beta,1}^*=h_{\rm{min}}=0$, $h_{\rm{max}}=1$;\\
        Set the precision $\epsilon=10^{-3}$ and iteration scaling factor $\gamma_{i}=10$.\\

        \WHILE {$\left(\frac{\partial P_{\rm{tr},1}(h_\beta)}{\partial h_\beta}|_{h_\beta=h_{\rm{max}}}\frac{\partial P_{\rm{tr},1}(h_\beta)}{\partial h_\beta}|_{h_\beta=h_{\rm{min}}}\right)\geq0$}\label{code:FindMaxh}
        	\STATE $h_{\rm{max}} = \gamma_{i} h_{\rm{max}}$;\\
        \ENDWHILE \\

        \WHILE {$\frac{\partial P_{\rm{tr},1}(h_\beta)}{\partial h_\beta}|_{h_\beta=h_{\beta,1}^*} \geq \epsilon$}\label{code:FindAccurateEta}
        	\STATE $h_{\beta,1}^*= \left( h_{\rm{max}} + h_{\rm{min}}\right)/2$;\\
            \IF  {$\frac{\partial P_{\rm{tr},1}(h_\beta)}{\partial h_\beta}|_{h_\beta=h_{\beta,1}^*}\geq0$}
            \STATE $h_{\rm{max}} = h_{\beta,1}^*$;
            \ELSE
            \STATE $h_{\rm{min}} = h_{\beta,1}^*$;
            \ENDIF
        \ENDWHILE \\
       \textbf{Output} ~ The optimal hovering altitude $h_\beta^* = R_{\beta} h_{\beta,1}^*$.\\          
    \end{algorithmic}
\end{algorithm}

\subsection{Optimal Coverage Radius of UAV-BSs}\label{Sec:Coverage4Area}

According to Lemma \ref{Lem:ScaleOptimalAltitude}, the optimal hovering altitudes of UAV-BSs are determined by the coverage radius. Then, $\mathcal{P}_1$ can be immediately rewritten as
\begin{flalign}\label{Prob:Optimal2DArrangement}
\mathcal{P}_{1\text{-}\rm{B}}: \,\,\min_{\left\{R_{\beta}\right\}} \,\,\, &\sum_{\beta=1}^{\kappa} \Phi_{\rm{st},\beta}(t)\\
\textrm{s.t.}\,\,\,&h_\beta^* = R_{\beta} h_{\beta,1}^*,\tag{\theequation a}\label{equ:Optimize2DArrangementa}\\
&R_{\beta} > 0.\tag{\theequation b}\label{equ:Optimize2DArrangementb}
\end{flalign}
\eqref{equ:Optimize2DArrangementa} is the optimal hovering altitude corresponding to $R_{\beta}$, and \eqref{equ:Optimize2DArrangementb} reflects the non-negativity of coverage radius. As shown in  \eqref{Prob:Optimal2DArrangement}, the static UAV-RF of considered area is the sum of individual static UAV-RF in each subregion. Hence, solving $\mathcal{P}_{1\text{-}\rm{B}}$ is equivalent to minimizing $\Phi_{\rm{st},\beta}(t)$ for all $\beta\in\{1,2,\cdots,\kappa\}$, separately. 

Substituting \eqref{Equ:UAVNumber} and  \eqref{Equ:ScaleTransmitPower} into \eqref{Equ:ChangingBatteryNum}, the static UAV-RF can be expressed as
\begin{equation}\label{Equ:PhiBetaExpression}
\begin{split}
\Phi_{\rm{st},\beta}(t) = \frac{S_\beta}{\pi E_{\rm{b}}} \bigg(\frac{P_{\rm{cu}}}{R^2_{\beta}} + \lambda_{\beta}\left( 2^{C/W} - 1 \right)P_{\rm{tr},1}(h_{\beta,1}^*)R^2_{\beta}\bigg).
\end{split}
\end{equation}
Thus, with the theoretical results shown in Section \ref{Sec:Coverage4SingleUAV}, the optimal placement of UAV-BSs that minimizes static UAV-RF in the considered area can be summarized as follows.
\begin{thm}\label{Thm:Optimal3DPlacement}
1) The optimal coverage radius and hovering altitude of UAV-BS in the $\beta$-th subregion at $t$ is
\begin{equation}\label{Equ:Optimal2DArrangement}
R^*_{\beta}(t) = \sqrt[4]{\frac{P_{\rm{cu}}}{\lambda_{\beta}(t)\left( 2^{C/W} - 1 \right)P_{\rm{tr},1}(h_{\beta,1}^*(t))}}
\end{equation}
and
\begin{equation}\label{Equ:OptimalHoveringAltitude}
h_\beta^*(t) = R^*_{\beta}(t) h_{\beta,1}^*(t),
\end{equation}
respectively. $h_{\beta,1}^*(t)$ is given by  \eqref{Equ:OptimalAltitude}.

2) The optimal placement of UAV-BSs in the $\beta$-th subregion is achieved when the on-board circuit power of one UAV-BS equals its transmit power. That is,
\begin{equation}\label{Equ:TransmitEqualToHovering}
P_{\rm{cu}} = \gamma_{\rm{tr}}^*(t) P_{\rm{tr}}(h_{\beta,1}^*(t)),
\end{equation}
where $\gamma_{\rm{tr}}^*(t)=\lambda_{\beta}(t) R^{*4}_{\beta}(t) \left( 2^{C/W} - 1\right)$.
\end{thm}
\begin{proof}
Denote the minimal static UAV-RF as $\Phi^*_{\rm{st},\beta}(t)$. Since both the on-board circuit power and user density in  \eqref{Equ:PhiBetaExpression} are positive, one has
\begin{equation}\label{Equ:InequalityOfPhi}
\Phi^*_{\rm{st},\beta}(t) = \frac{2S_\beta}{\pi E_{\rm{b}}} \sqrt{\lambda_{\beta}\left( 2^{C/W} - 1 \right)P_{\rm{cu}}P_{\rm{tr},1}(h_{\beta,1}^*)},
\end{equation}
and $\Phi^*_{\rm{st},\beta}(t)$ is achieved when $R_{\beta}(t) = R^*_{\beta}(t)$, which is shown in  \eqref{Equ:Optimal2DArrangement}. According to Lemma \ref{Lem:ScaleOptimalAltitude}, the optimal hovering altitude corresponding to $R^*_{\beta}(t)$ is given by  \eqref{Equ:OptimalHoveringAltitude}. Substitute  \eqref{Equ:Optimal2DArrangement} into \eqref{Equ:ScaleTransmitPower}, \eqref{Equ:TransmitEqualToHovering} can be easily proved.
\end{proof}
Theorem \ref{Thm:Optimal3DPlacement} not only provides the optimal hovering altitude and the optimal coverage radius, but also points out the optimal transmit power and circuit power allocation. The former result provides valuable insights for UAV-BSs positioning while the latter result presents an efficient resource allocation for multi-function usage of UAV-BSs. The physical meaning of Theorem \ref{Thm:Optimal3DPlacement} is intelligible. When on-board circuit power is high, large coverage radius can decrease the number of active UAV-BSs. According to  \eqref{Equ:ChangingBatteryNum}, small $N_\beta(t)$ decreases the effects of high $P_{\rm{cu}}$ on static UAV-RF, and hence $\Phi_{\rm{st},\beta}(t)$ is reduced. By contrast, when on-board circuit power is low, small coverage radius can decrease transmit power, which also decreases $\Phi_{\rm{st},\beta}(t)$. Specifically, when $P_{\rm{cu}}=0$, we have $R^{*}_{\beta}=0$ and $h_\beta^*=0$. That is, users can connect to UAV-BSs just at their balance positions. In this way, enlarging the number of UAVs does not consume any on-board circuit power, while the transmit power is saved.

\section{Dynamic Placement of UAV-BSs in Considered Duration}\label{Sec:DynamicCoverage}

This section shall consider the placement of UAV-BSs in time-dimension. As illustrated in Section \ref{sec:learnedDensityPattern}, there exists a trade-off between updating the placement of UAV-BSs and reducing the dynamic UAV-RF. In this case, we firstly formulate the optimal placement updating problem in terms of minimum dynamic UAV-RF. We demonstrate that the original problem is NP-hard and can not be solved in conventional manners. Finally, a sequential method is proposed to update the UAV-BS placement near-optimally in polynomial time.

\subsection{Formulation of the Optimal Placement Updating}\label{sec:FormulationOfOptimalUpdate}

Recall that only at beginning of time-slots, the placement of UAV-BSs can be updated to be optimal. Denote the update epochs of UAV-BSs in the considered area as $\tau_{i}$ ($i=0,\cdots,N_{\tau}$). $N_{\tau}$ is the number of update epochs and $0 \leq N_{\tau}\leq {T}/{\mu}$. Specifically, let $\tau_0=0$. Denote the updated coverage radius of UAV-BSs in the $\beta$-th subregion at $\tau_i$ as $R_{\beta}(\tau_{i})$. The corresponding number of UAV-BSs can be expressed as $N_\beta(\tau_{i}) = S_\beta/R_{\beta}(\tau_{i})$. Furthermore, if $N_\tau\geq1$, the number of UAV-BSs need re-positioning at $\tau_i$ can be given by 
\begin{equation}
\zeta_\tau|_{\tau_{i-1}}^{\tau_{i}} = \max\left(N(\tau_{i-1}), N(\tau_{i})\right),
\end{equation}
where $N(\tau_i)=\sum_{\beta=1}^{\kappa}N_\beta(\tau_{i})$ and $i=1,\cdots,N_\tau$.

Then, the mobility energy of UAV-BSs at $\tau_{i}$ ($i=1,\cdots,N_\tau$) can be expressed as\cite{di2015energy}
\begin{equation}\label{Equ:MigrationEnergy}
\Omega_{\rm{m}}(\tau_{i}) = \sum_{l=1}^{\zeta_\tau|_{\tau_{i-1}}^{\tau_{i}}}P_{\rm{h}}\frac{d_\beta(\tau_{i},l)}{v_{\rm{h}}} + I(\Delta h_\beta(\tau_{i},l)) P_{\rm{a}}\frac{\Delta h_\beta(\tau_{i},l)}{v_{\rm{a}}} - \big(1-I(\Delta h_\beta(\tau_{i},l))\big) P_{\rm{d}}\frac{\Delta h_\beta(\tau_{i},l)}{v_{\rm{d}}}, 
\end{equation}
where $P_{\rm{h}}$ and $v_{\rm{h}}$ are the mobility power and velocity for one single UAV-BS in the horizontal direction, respectively. $d_\beta(\tau_{i},l)$ and $\Delta h_\beta(\tau_{i},l)$ are the horizontal moving distance and the variation of the height of the $l$-th UAV-BS in the $\beta$-th subregion at $\tau_{i}$, respectively. $P_{\rm{a}}$ and $v_{\rm{a}}$ denote the ascending power and ascending velocity of one single UAV-BS, respectively. Similarly, $P_{\rm{d}}$ and $v_{\rm{d}}$ denote the descending power and descending velocity. $I(\Delta h_\beta(\tau_{i},l))$ is the indicative function, i.e.
\begin{equation}
I(\Delta h_\beta(\tau_{i},l)) = 
\begin{cases}
1, \,\,\Delta h_\beta(\tau_{i},l) \geq 0\\
0, \,\,\Delta h_\beta(\tau_{i},l) < 0.
\end{cases}
\end{equation}
Specifically, $\Omega_{\rm{m}}(\tau_0)=0$. Note that $\Omega_{\rm{m}}(\tau_{i})$ ($i=1,\cdots,N_\tau$) is determined by the total moving distance of UAV-BSs at $\tau_i$, related to the placement of UAV-BSs at $\tau_{i-1}$. 

We consider the optimal updating strategy as $\{\tau_0,\cdots,\tau_{N_{\tau}}\}$. With \eqref{Equ:MigrationEnergy}, the average dynamic UAV-RF in considered duration can be given by 
\begin{equation}\label{Equ:DynamicUAV_RF}
\begin{split}
\bar{\Phi}_{\rm{dn}}|^{T}_0 = \frac{1}{T} \int_0^T \sum_{\beta=1}^{\kappa}\Phi_{\rm{dn},\beta}(t) d t =\frac{1}{T} \left\{ \int_0^T \sum_{\beta=1}^{\kappa}\Phi_{\rm{st},\beta}(t) d t + \frac{1}{E_{\rm{b}}}\sum_{i=0}^{N_\tau}\Omega_{\rm{m}}(\tau_i)\right\},
\end{split}
\end{equation}
where $\Phi_{\rm{dn},\beta}(t)$ is the instantaneous dynamic UAV-RF at $t$. The corresponding optimal placement updating method can be immediately expressed as
\begin{flalign}\label{Equ:MinAverageUAV-RF}
\mathcal{P}_2: \,\,\min_{\{\tau_0,\cdots,\tau_{N_{\tau}}\}} \,\,\, &\bar{\Phi}_{\rm{dn}}|^{T}_0\\
\textrm{s.t.}\,\,\,& 0 \leq N_{\tau}\leq {T}/{\mu},\tag{\theequation a}\label{equ:MainOptimizeProblemConstraintsc}\\
&\tau_0=0,~\tau_{N_{\tau}} < T,\tag{\theequation b}\label{equ:MainOptimizeProblemConstraintsa}\\
& \text{if}~N_\tau \geq 1, \, \tau_{i} < \tau_{i+1} \, (i=0,\cdots,N_\tau-1). \tag{\theequation c}\label{equ:MainOptimizeProblemConstraintsb}
\end{flalign}
 \eqref{equ:MainOptimizeProblemConstraintsc} shows the bounds of update epochs.  \eqref{equ:MainOptimizeProblemConstraintsa} and  \eqref{equ:MainOptimizeProblemConstraintsb} denotes the range and the order of update epochs, respectively. 

Problem $\mathcal{P}_2$ is a typical multi-step decision process. Because the objective function given by  \eqref{Equ:MigrationEnergy} is nonlinear, $\mathcal{P}_2$ can be solved by nonlinear dynamic programming (NLDP) method. However, because the number of updating times is unknown in advance, the number of recursion formula can not be determined as well. Hence, $\mathcal{P}_2$ is NP-hard and difficult to be directly solved as the size of problem grows. Hereinafter, to efficiently find the near optimal solution of $\mathcal{P}_2$, we propose a new sequential method with polynomial computational complexity.

\subsection{Trajectory Planning}\label{sec:TrajPlan}

Given the update epochs $\tau_i$ ($i=1,\cdots,N_\tau$), the mobility energy shown in  \eqref{Equ:MigrationEnergy} is only determined by the updating trajectory of the placement of UAV-BSs. Therefore, we firstly analyze the optimal trajectory planning method that minimizes the mobility energy of UAV-BSs.

Let the 3D position of the $k$-th UAV-BS at $\tau_{i-1}$ be $\psi(k,\tau_{i-1})$, where $k=1,\cdots,N(\tau_{i-1})$. Then, the corresponding position set of UAV-BSs can be expressed as
\begin{equation}
\Psi(\tau_{i-1}) = \left( \psi(1,\tau_{i-1}), \cdots, \psi(N(\tau_{i-1}),\tau_{i-1}) \right),~i=1,\cdots,N_\tau.
\end{equation}
To balance the number of UAV-BSs before and after updating the placement of UAV-BSs at $\tau_i$, the redundant UAV-BSs should be recalled to RSC when $N(\tau_{i-1}) > N(\tau_i)$, whereas the additional UAV-BSs should be supplemented by RSC when $N(\tau_{i-1}) < N(\tau_i)$. Let the position of RSC be $\psi_{\rm{R}}(k)$. The position set of recalled or supplemented UAV-BSs can be expressed as  \eqref{Equ:RecallAndSupplenmentedUAVs}, which is shown at the top of this page. Then, the position sets of UAV-BSs at $\tau_{i-1}$ and $\tau_{i}$ can be immediately given by

\begin{figure*}[!t]
\normalsize
\begin{equation}\label{Equ:RecallAndSupplenmentedUAVs}
\Psi_{\rm{R}} = \begin{cases}
\left( \psi_{\rm{R}}(1), \cdots, \psi_{\rm{R}}(|N(\tau_i) - N(\tau_{i+1})|) \right), ~&\text{if~} N(\tau_i) \neq N(\tau_{i+1})\\
\quad\quad\quad\quad\quad\,\,\,\,\varnothing, ~&\text{if~} N(\tau_i) = N(\tau_{i+1}).
\end{cases}
\end{equation}
\hrulefill
\vspace*{4pt}
\end{figure*}

\begin{equation}
\hat{\Psi}(\tau_{i-1}) = \begin{cases}
\Psi(\tau_{i-1}), ~&\text{if~} N(\tau_{i-1}) \geq N(\tau_{i})\\
\left( \Psi(\tau_{i-1}), \Psi_{\rm{R}} \right), ~&\text{if~} N(\tau_{i-1}) < N(\tau_{i})
\end{cases}
\end{equation}
and
\begin{equation}
\hat{\Psi}(\tau_{i}) = \begin{cases}
\left( \Psi(\tau_{i}), \Psi_{\rm{R}} \right), ~&\text{if~} N(\tau_{i-1}) \geq N(\tau_{i})\\
\Psi(\tau_{i}), ~&\text{if~} N(\tau_{i-1}) < N(\tau_{i}),
\end{cases}
\end{equation}
$i=1,\cdots,N_\tau$. Obviously, the volume of both the position sets $\hat{\Psi}(\tau_{i-1})$ and $\hat{\Psi}(\tau_{i})$ is $\zeta_\tau|_{\tau_{i-1}}^{\tau_{i}}$. Our goal is to find the optimal mapping between $\hat{\Psi}(\tau_{i-1})$ and $\hat{\Psi}(\tau_{i})$ that minimizes $\Omega_{\rm{m}}(\tau_{i})$. 

According to \eqref{Equ:MigrationEnergy}, one can see that finding the optimal mapping is a standard integer linear programming (ILP). In general, this problem can be solved by using standard ILP solution methods. However, these solutions may not be efficient as the size of the problem grows. Due to the potential high number of UAV-BSs, a more efficient technique is needed. To this end, we transform the trajectory planning problem into a standard assignment problem, which can be solved with the Hungarian method in polynomial time $O\left(\left(\zeta_\tau|_{\tau_i}^{\tau_{i+1}}\right)^3 \right)$\cite{kuhn1955hungarian}.

Define the mobility energy matrix as $\mathbf{C}$, which is a $\zeta_\tau|_{\tau_{i-1}}^{\tau_{i}} \times \zeta_\tau|_{\tau_{i-1}}^{\tau_{i}}$ square matrix. Let $\hat{\psi}_{x,y}(k,\tau_{i})$ and $\hat{\psi}_{z}(k,\tau_{i})$ be the horizontal and vertical coordinate of the $k$-th element of $\hat{\psi}(k,\tau_{i})$, respectively. Then, the element at the $k$-th row and the $l$-th column of $\mathbf{C}$ can be immediately written as 
\begin{equation}\label{Equ:CostMat}
\begin{split}
\mathbf{C}(k,l) = &P_{\rm{h}}\frac{||\hat{\psi}_{x,y}(l,\tau_{i}) - \hat{\psi}_{x,y}(k,\tau_{i-1})||_2}{v_{\rm{h}}} + \\
& \quad\quad\quad I(\hat{\psi}_{z}(l,\tau_{i}) - \hat{\psi}_{z}(k,\tau_{i-1})) P_{\rm{a}}\frac{\hat{\psi}_{z}(l,\tau_{i}) - \hat{\psi}_{z}(k,\tau_{i-1})}{v_{\rm{a}}} - \\
& \quad\quad\quad \quad\quad\quad \big(1-I(\hat{\psi}_{z}(l,\tau_{i}) - \hat{\psi}_{z}(k,\tau_{i-1}))\big) P_{\rm{d}}\frac{\hat{\psi}_{z}(l,\tau_{i}) - \hat{\psi}_{z}(k,\tau_{i-1})}{v_{\rm{d}}}.\\
\end{split}
\end{equation}
$\mathbf{C}(k,l)$ represents the mobility energy when a UAV-BS moves from $\hat{\psi}(k,\tau_{i-1})$ to $\hat{\psi}(l,\tau_{i})$. Denote the corresponding $\zeta_\tau|_{\tau_{i-1}}^{\tau_{i}} \times \zeta_\tau|_{\tau_{i-1}}^{\tau_{i}}$ assignment matrix as $\mathbf{Z}$, where the element $\mathbf{Z}(k,l)$ is 1 if the UAV-BS at $\hat{\psi}(k,\tau_{i-1})$ is assigned to $\hat{\psi}(l,\tau_{i})$, or 0 otherwise. In this way, the standard assignment problem on finding the optimal trajectory can be formulated as
\begin{flalign}\label{Equ:OptimalTrajectoryPlanning}
\mathcal{P}_{2\text{-}\rm{A}}: \,\,\min_{\mathbf{Z}} \,\,\, &\sum_{k=1}^{\zeta_\tau|_{\tau_{i-1}}^{\tau_{i}}} \sum_{l=1}^{\zeta_\tau|_{\tau_{i-1}}^{\tau_{i}}} \mathbf{C}(k,l) \mathbf{Z}(k,l)\\
\textrm{s.t.}\,\,\,&\sum_{k=1}^{\zeta_\tau|_{\tau_{i-1}}^{\tau_{i}}} \mathbf{Z}(k,l)=1,~\sum_{l=1}^{\zeta_\tau|_{\tau_{i-1}}^{\tau_{i}}} \mathbf{Z}(k,l)=1.\label{equ:OptimalTrajectoryPlanning}
\end{flalign}
Constraint  \eqref{equ:OptimalTrajectoryPlanning} guarantees that the assignment between $\hat{\Psi}(\tau_{i-1})$ and $\hat{\Psi}(\tau_{i})$ is one-to-one. 

\subsection{Placement Updating Strategy}

Following the optimal trajectory planning method, this part gives the optimal updating strategy $\{ \tau_0,\cdots,\tau_{N_\tau} \}$ by solving problem $\mathcal{P}_{2}$. To this end, we first illustrate the Markov nature of $\mathcal{P}_{2}$, as shown in the following Lemma.

\begin{lem}\label{lem:Markov}
The static UAV-RF at $t$ can be expressed as
\begin{equation}\label{equ:MarkovEnergy}
\Phi_{\rm{st},\beta|(\tau_0,\cdots,\tau_i)}(t) = \Phi_{\rm{st},\beta|\tau_i}(t),
\end{equation}
where $\tau_i = \max\{ \tau_i|\tau_i\leq t \}$, $i=0,\cdots,N_\tau$. $\Phi_{\rm{st},\beta|(\tau_0,\cdots,\tau_i)}(t)$ and $\Phi_{\rm{st},\beta|\tau_i}(t)$ are the corresponding static UAV-RFs at time $t$ when $(\tau_0,\cdots,\tau_i)$ and $\tau_i$ are given, respectively.
\end{lem}
\begin{proof}
Since the placement UAV-BSs has been updated at $\tau_i$, the coverage radius at time $t$ can be given by $R_{\beta}(t)=R_{\beta}(\tau_i)$. In addition, because as shown in  \eqref{Equ:PhiBetaExpression}, the static UAV-RF is only determined by the coverage radius $R_{\beta}(t)$, \eqref{equ:MarkovEnergy} follows.
\end{proof}
With Lemma \ref{lem:Markov}, we can see that the static UAV-RF at time $t$ is determined by the nearest update epoch in a Markov mode. This motivates us to investigate the average dynamic UAV-RF between consecutive update epochs, denoted by
\begin{equation}
\bar{\Phi}_{\rm{dn}}|^{\tau_{i}}_{\tau_{i-1}} = \frac{ \int_{\tau_{i-1}}^{\tau_{i}} \sum_{\beta=1}^{\kappa}\Phi_{\rm{st},\beta}(t) d t + \Omega_{\rm{m}} (\tau_{i}) }{\tau_{i}-\tau_{i-1}},~i=1,\cdots,N_\tau.
\end{equation}
Considering the fact that the number of update epochs is unknown, we can sequentially decide the update epochs by minimizing $\bar{\Phi}_{\rm{dn}}|^{\tau_{i}}_{\tau_{i-1}}$. The following Theorem established the Pareto-optimality of this sequential decision method. 

\begin{thm}\label{thm:ParetoOptimalAlg}
Given $\tau_0=0$, the sequential-Markov-greedy-decision (S-MGD) method can be expressed as
\begin{equation}\label{equ:paretoDecision}
\tau_{i} = \argmin_{\tau_{i}\in(\tau_{i-1},T)} \bar{\Phi}_{\rm{dn}}|^{\tau_{i}}_{\tau_{i-1}},~i=1,\cdots,N_\tau,
\end{equation}
and the corresponding $\{ \tau_0,\cdots,\tau_{N_\tau} \}$ is a Pareto-optimal placement updating strategy.
\end{thm}
\begin{proof}

\begin{figure}[htbp]
\centering{}
\includegraphics[width=0.6\textwidth]{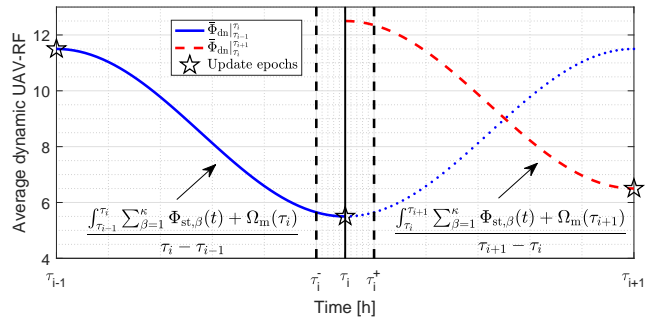}
\caption{A general case of proposed S-MGD method. The optimal update epochs are marked with stars.}\label{Fig:ParetoSample}
\end{figure}

For the convenience of proof, we utilize a general case shown in Fig. \ref{Fig:ParetoSample} to demonstrate the Pareto-optimality of proposed S-MGD method. The solid line and dotted line illustrate how $\bar{\Phi}_{\rm{dn}}|_{\tau_{i-1}}^{\tau_i}$ varies with respect to $\tau_i$, given $\tau_{i-1}$. According to \eqref{equ:paretoDecision}, the optimal update epoch $\tau_i$ is greedily determined and is marked with star in Fig. \ref{Fig:ParetoSample}. That is, $\bar{\Phi}_{\rm{dn}}|_{\tau_{i-1}}^{\tau_i}$ achieves its minimum at the optimal $\tau_i$. Similarly, the red dashed line depicts how $\bar{\Phi}_{\rm{dn}}|_{\tau_{i}}^{\tau_{i+1}}$ varies with $\tau_{i+1}$ when $\tau_i$ is given by  \eqref{equ:paretoDecision}. Also, the optimal $\tau_{i+1}$ that minimizes $\bar{\Phi}_{\rm{dn}}|_{\tau_{i}}^{\tau_{i+1}}$ is marked with star in Fig. \ref{Fig:ParetoSample}. Then, we prove the Pareto-optimality of Theorem \ref{thm:ParetoOptimalAlg}. That is, for the average dynamic UAV-RF between each consecutive update epochs during $[0,T]$, we can not decrease one of them without increasing the others. 

Take one of the the duration $[\tau_{i-1},\tau_{i+1}]$ ($i=1,\cdots,N_\tau-1$) for example. According to Lemma \ref{lem:Markov}, adjusting $\tau_i$ only affects $\bar{\Phi}_{\rm{dn}}|_{\tau_{i-1}}^{\tau_{i}}$ and $\bar{\Phi}_{\rm{dn}}|_{\tau_{i}}^{\tau_{i+1}}$. Let $\tau^{-}_{i}$ and $\tau^{+}_{i}$ denote the time belonging to $(\tau_{i-1},\tau_{i})$ and $(\tau_{i},\tau_{i+1})$, respectively. As shown in Fig. \ref{Fig:ParetoSample}, due to $\tau_i$ minimizes $\bar{\Phi}_{\rm{dn}}|^{\tau_{i}}_{\tau_{i-1}}$, one has
\begin{equation}
\begin{cases}
\bar{\Phi}_{\rm{dn}}|^{\tau^{-}_{i}}_{\tau_{i-1}} > \bar{\Phi}_{\rm{dn}}|^{\tau_{i}}_{\tau_{i-1}}\\
\bar{\Phi}_{\rm{dn}}|^{\tau^{+}_{i}}_{\tau_{i-1}} > \bar{\Phi}_{\rm{dn}}|^{\tau_{i}}_{\tau_{i-1}}.
\end{cases}
\end{equation}
That is, decreasing $\bar{\Phi}_{\rm{dn}}|^{\tau_{i+1}}_{\tau_{i}}$ without increasing $\bar{\Phi}_{\rm{dn}}|^{\tau_{i}}_{\tau_{i-1}}$ is impossible. This result also holds for any interval $[\tau_{i-1},\tau_{i+1}]$ ($i=1,\cdots,N_\tau-1$). Hence, we conclude that our proposed S-MGD based method Pareto-optimally solves $\mathcal{P}_{2}$. This completes the proof.
\end{proof}


\begin{rem}\label{Rem:Complexity}
It is noteworthy that the proposed S-MGD method has polynomial computational complexity, bounded by 
\begin{equation}\label{equ:Complexity}
O\left(T\zeta_{\tau,\min}^3/\mu \right) \leq \Theta_{\rm{S}} \leq O\left(T^2\zeta_{\tau,\max}^3/\mu^2 \right).
\end{equation}
$\zeta_{\tau,\min}$ and $\zeta_{\tau,\max}$ are the minimal and maximal number of UAV-BSs in considered duration.
\end{rem}
\begin{proof}
The number of time-slots during $(0,T)$ is $T/\mu$. Accordingly, the number of greedy searching shown in  \eqref{equ:paretoDecision} is bounded by $T/\mu$ and ${T(T/\mu+1)}/{(2\mu)}$. Note that in each of these greedy searchings, the optimal trajectory needs to be calculated by Hungarian method with a complexity of $O\left(\left(\zeta_\tau|_{\tau_i}^{\tau_{i+1}}\right)^3 \right)$\cite{kuhn1955hungarian}. Since $\zeta_{\tau,\min} \leq \zeta_\tau|_{\tau_i}^{\tau_{i+1}} \leq \zeta_{\tau,\max}$, \eqref{equ:Complexity} can be derived.
\end{proof}

\begin{rem}\label{rem:future}
The sequential decision process shown in \eqref{equ:paretoDecision} is based on the global knowledge of ground user density in $(0,T)$, which is inferred by the pattern formation systems.
\end{rem}
Following Remark \ref{rem:future}, we analyze the effects of density pattern accuracy on the performance of proposed S-MGD method, as shown in following section.

\section{Learning the Density Pattern of Ground Users}\label{sec:insight}
This section shall focus on predicting the user density by machine learning techniques. In particular, we first explicitly analyze the effect of pattern formation accuracy on the increment of UAV-RF. Then, we characterize the feature of subregions that sufficient sampling is required for reducing UAV-RF, based on the relation between sampling number and pattern accuracy.


\subsection{Effects of Pattern Formation Accuracy}\label{Sec:ImportantIssues}

Recall that the minimum static UAV-RF is achieved at $R^*_{\beta}$, based on specific density of ground users (conf. \eqref{Equ:OptimalHoveringAltitude}). An inaccurate predicted density of users, denoted as $\hat{\lambda}_{\beta}$, would increase static UAV-RF and reduce the life-time of batteries. Let the static UAV-RF corresponding to $\hat{\lambda}_{\beta}$ be
\begin{equation}
\begin{split}
\hat{\Phi}_{\rm{st},\beta}(t) = \frac{S_\beta}{\pi E_{\rm{b}}} \bigg(\frac{P_{\rm{cu}}}{\hat{R}^{*2}_{\beta}} + \lambda_{\beta}\left( 2^{C/W} - 1 \right)P_{\rm{tr},1}(h_{\beta,1}^*)\hat{R}^{*2}_{\beta} \bigg),
\end{split}
\end{equation}
where $\hat{R}^{*}_{\beta}$ is the generated coverage radius with $\hat{\lambda}_{\beta}$, i.e.
\begin{equation}
\hat{R}^{*}_{\beta} = \sqrt[4]{\frac{P_{\rm{cu}}}{\hat{\lambda}_{\beta}\left( 2^{C/W} - 1 \right)P_{\rm{tr},1}(h_{\beta,1}^*)}}.
\end{equation}
The following theorem characterizes the average increased UAV-RF caused by inaccurate prediction on $\lambda_\beta$.
\begin{thm}\label{thm:biasAndVarOnStatic}
The expected value of increased static UAV-RF is
\begin{equation}\label{equ:avgIncreasePhi}
\Delta \Phi_{\rm{st},\beta}(t) \triangleq \mathbb{E}\left[ \hat{\Phi}_{\rm{st},\beta}(t) - \Phi_{\rm{st},\beta}(t) \right] = \Lambda_\beta(t) \xi_\beta(t).
\end{equation}
$\Lambda_\beta(t)$ is the eigenvalue of the $\beta$-th subregion at time $t$, given by
\begin{equation}\label{equ:EigenValueOfSubregion}
\begin{split}
&\Lambda_\beta(t) = \frac{S_\beta}{4\pi E_{\rm{b}}} \sqrt{P_{\rm{cu}}\left( 2^{C/W} - 1 \right)P_{\rm{tr},1}(h_{\beta,1}^*)/ \lambda^3_{\beta}}.
\end{split}
\end{equation}
$\xi_\beta(t) \triangleq \text{var}(\hat{\lambda}_{\beta}(t)) + \text{bias}^2(\lambda_{\beta}(t))$ is the generalization error with quadratic loss function, which is composed by the variance and the square of bias of predicted ground user density. 
\end{thm}
\begin{proof}

Since ${R}^{*}_{\beta}$ minimizes $\Phi_{\rm{st},\beta}(t)$, ${\partial \Phi_{\rm{st},\beta}(t)}/{\partial R_{\beta}}|_{{R}^{*}_{\beta}}=0$ can be derived. Using Taylor series expansion at ${R}^{*}_{\beta}$, $\hat{\Phi}_{\rm{st},\beta}(t)$ can be rewritten as
\begin{equation}\label{equ:TaylorPhi}
\begin{split}
\hat{\Phi}_{\rm{st},\beta}(t) = \Phi^*_{\rm{st},\beta}(t) + \frac{1}{2}\frac{\partial^2 \Phi_{\rm{st},\beta}(t)}{\partial^2 R_{\beta}}\Big|_{{R}^{*}_{\beta}} \left[ \hat{R}^{*}_{\beta} - {R}^{*}_{\beta} \right]^2 + O\left[ \left( \hat{R}^{*}_{\beta} - {R}^{*}_{\beta} \right)^3 \right].
\end{split}
\end{equation}
Similarly, $\hat{R}^{*}_{\beta}(t)$ can be written as
\begin{equation}\label{equ:TaylorR}
\hat{R}^{*}_{\beta}(t) = {R}^{*}_{\beta} + \frac{\partial {R}_{\beta}}{\partial \lambda_{\beta}}\Big|_{\lambda_{\beta}}\left( \hat{\lambda}_{\beta} - \lambda_{\beta}\right) + O\left( \left( \hat{\lambda}_{\beta} - \lambda_{\beta}\right)^2 \right). 
\end{equation}
Substituting $\frac{\partial^2 \Phi_{\rm{st},\beta}(t)}{\partial^2 R_{\beta}(t)}$ and $\frac{\partial {R}_{\beta}(t)}{\partial \lambda_{\beta}(t)}$ into  \eqref{equ:TaylorPhi} and \eqref{equ:TaylorR}, one has
\begin{equation}\label{equ:biasAndVarOnStatic}
\begin{split}
&\mathbb{E}\left[\hat{\Phi}_{\rm{st},\beta}(t) - \Phi_{\rm{st},\beta}(t) \right] = \Lambda_\beta(t) \mathbb{E}\left[\left(\hat{\lambda}_{\beta}-\lambda_{\beta}\right)^2\right],
\end{split}
\end{equation}
where $\Lambda_\beta(t)$ is given by \eqref{equ:EigenValueOfSubregion} and $\mathbb{E}\left[\left(\hat{\lambda}_{\beta}-\lambda_{\beta}\right)^2\right]$ is the generalization error defined as\cite{goodfellow2016deep}
\begin{equation}\label{equ:MSE_biasAndVar}
\begin{split}
\xi_\beta(t) \triangleq \mathbb{E}\left[ \left( \hat{\lambda}_{\beta} - \lambda_{\beta} \right)^2 \right]&= \mathbb{E}\left\{\left[ \left( \hat{\lambda}_{\beta} - \mathbb{E}(\hat{\lambda}_{\beta}) \right) + \left( \mathbb{E}(\hat{\lambda}_{\beta}) - \lambda_{\beta} \right) \right]^2\right\}\\
&=\text{var}(\hat{\lambda}_{\beta}) + \left[\mathbb{E}(\hat{\lambda}_{\beta}) - \lambda_{\beta}\right]^2\\
&=\text{var}(\hat{\lambda}_{\beta}) + \text{bias}^2(\lambda_{\beta}).
\end{split}
\end{equation}
$\mathbb{E}(\cdot)$ is the expected value of $(\cdot)$. Substituting \eqref{equ:MSE_biasAndVar} into  \eqref{equ:biasAndVarOnStatic}, \eqref{equ:avgIncreasePhi} can be proved.
\end{proof}

Following Theorem \ref{thm:biasAndVarOnStatic}, one can see that the generalization error of pattern formation systems proportionally contribute to the expected value of increased static UAV-RF. Also, observing the eigenvalue $\Lambda_\beta(t)$ shown in \eqref{equ:EigenValueOfSubregion}, we conclude that in subregions with small $\lambda_\beta$, $\Delta \Phi_{\rm{st},\beta}(t)$ is more sensitive to generalization error. That is, when the time-varying density of ground users in one specific subregion is low, the generalization error should be strictly restrained.

\subsection{Bounds on Generalization Error and Sampling Number}

Discussing the upper-bound of generalization error and the corresponding lower-bound of sampling number under specific constraints is one of the core problems in the literature of machine learning, since a large sample size leads to a lot of work\cite{goodfellow2016deep,dobrusin1958statistical}. Therefore, we consider minimizing the total sampling number of system under the constraint of maximal static UAV-RF increment of considered area.

According to Vapnik-Chervonenkis (VC) theorem\cite{vapnik2013nature}, the generalization error convergences with respect to increasing sample size, i.e.
\begin{equation}\label{equ:VCDim}
\xi_\beta(t) \leq \hat{\xi}_\beta(t) + \varepsilon(d,N_{\rm{s},\beta}(t),\delta)
\end{equation}
holds true with probability at least $1-\delta$, where
\begin{equation}
\varepsilon(d,N_{\rm{s},\beta}(t),\delta) = \sqrt{\frac{1}{2N_{\rm{s},\beta}(t)}\left( \log d + \log\frac{1}{\delta} \right)},
\end{equation}
$\hat{\xi}_\beta(t)$ is the training error, $d$ is the volume of hypothesis space and $\delta \in (0,1)$. $N_{\rm{s},\beta}(t)$ is the sampling number in the $\beta$-th subregion at time $t$. Considering the fact that training error is determined by the capacity of learning algorithm\cite{goodfellow2016deep}, we let $\xi_{\max}$ be the maximal training error and  \eqref{equ:VCDim} can be rewritten as
\begin{equation}\label{equ:VCDimUB}
\xi_\beta(t) \leq \xi_{\max} + \varepsilon\left(d,N_{\beta,\rm{s}}(t),\delta\right).
\end{equation}


Denote the maximal tolerable static UAV-RF increment of considered area as $\Delta \Phi_{\rm{st},\max}(t)$. Then, the optimization problem on minimizing the sampling number of considered area can be expressed as follows.
\begin{flalign}\label{Equ:MinSamplingNum}
\mathcal{P}_3: \,\,\min_{\left\{\xi_\beta(t)\right\}} \,\,\, &\sum_{\beta=1}^{\kappa} N_{\beta,\rm{s}}(t)\\
\textrm{s.t.}\,\,\,&\xi_\beta(t) \geq 0,\tag{\theequation a}\label{equ:SNOptimizeProblemConstraintsa}\\
&  \sum_{\beta=1}^{\kappa} \Delta\Phi_{\rm{st},\beta}(t) \leq \Delta \Phi_{\rm{st},\max}(t). \tag{\theequation b}\label{equ:SNOptimizeProblemConstraintsb}
\end{flalign}
\eqref{equ:SNOptimizeProblemConstraintsa} denotes the non-negativity of generalization error, and \eqref{equ:SNOptimizeProblemConstraintsb} shows the upper-bound of the average increment of static UAV-RF in considered area at time $t$. By solving $\mathcal{P}_3$, the minimal sampling number and corresponding maximal generalization error of each subregion in considered area are presented as Proposition \ref{lem:minSampleNumber}.
\begin{prop}\label{lem:minSampleNumber}
With maximal tolerable increment of static UAV-RF in considered area $\Delta \Phi_{\rm{st},\max}(t)$, the minimal sampling number in the $\beta$-th subregion at $t$ is 
\begin{equation}\label{equ:SNOfBeta}
N_{\beta,\rm{s}}(t) = \omega^2 \Lambda_\beta^2(t) \left( \log d - \log \delta \right)/8,
\end{equation}
where
\begin{equation}\label{equ:omega}
\omega = \frac{2\kappa}{\Delta \Phi_{\rm{st},\max}(t) - \xi_{\max}\sum_{\beta=1}^{\kappa}\Lambda_\beta(t)}.
\end{equation}
Correspondingly, the maximum generalization error of pattern formation system is
\begin{equation}\label{equ:MinGeneralizationErr}
\xi_\beta(t) \leq \xi_{\max} + \frac{2}{\omega \Lambda_\beta(t)}.
\end{equation}
\end{prop}
\begin{proof}
The Lagrangian function of $\mathcal{P}_3$ is
\begin{equation}\label{equ:LagrangianFunction}
\begin{split}
\mathcal{L}\left(\sum_{\beta=1}^{\kappa} N_{\beta,\rm{s}}(t), \omega \right) = \sum_{\beta=1}^{\kappa} &N_{\beta,\rm{s}}(t) - \omega \left( \sum_{\beta=1}^{\kappa} \Delta \Phi_{\rm{st},\beta}(t) - \Delta \Phi_{\rm{st},\max}(t) \right).
\end{split}
\end{equation}
Let ${\partial \mathcal{L}\left(\sum_{\beta=1}^{\kappa} N_{\beta,\rm{s}}(t), \omega \right)}/{\partial N_{\beta,\rm{s}}(t)}=0$, we have
\begin{equation}
2 = \omega \Lambda_\beta(t) \sqrt{\left( \log d - \log \delta \right)/\left(2N_{\beta,\rm{s}}(t)\right)},
\end{equation}
which is equivalent to \eqref{equ:SNOfBeta}. Since $\omega$ is the positive value that subjects to $\sum_{\beta=1}^{\kappa} \Delta\Phi_{\rm{st},\beta}(t) = \Delta \Phi_{\rm{st},\max}(t)$, \eqref{equ:omega} can be derived. Also, one can obtain \eqref{equ:MinGeneralizationErr} by substituting \eqref{equ:SNOfBeta} into \eqref{equ:VCDimUB}.
\end{proof}
With Proposition \ref{lem:minSampleNumber}, we can conclude that in our considered area, subregions with high eigenvalue need large sampling number. Further, let $\kappa=1$, Proposition \ref{lem:minSampleNumber} implies that
\begin{equation}\label{equ:SampleNumberSingleRegion}
N_{\rm{s},1}(t) = \frac{\Lambda_\beta^2(t)\left( \log d - \log \delta \right)}{2\left( \Delta \Phi_{\rm{st},\max}(t) - \xi_{\max} \Lambda_\beta(t) \right)^2 },
\end{equation}
and
\begin{equation}
\xi_1(t) \leq {\Delta \Phi_{\rm{st},\max}(t)}/{\Lambda_\beta(t)}.
\end{equation}
This agrees with our previous result that $\Delta \Phi_{\rm{st},\beta}(t)$ in subregions with small $\lambda_\beta$ is more sensitive to generalization error, as shown in Theorem \ref{thm:biasAndVarOnStatic}.


\begin{table}[tbp]
\small
 \caption{Common parameters used in numerical analyses}\label{tab:ComParameters}
 \centering
 \begin{tabular}{ccccc}
  \toprule
  urban $(a,b,\eta_0,\eta_1)$ & dense urban $(a,b,\eta_0,\eta_1)$& suburban $(a,b,\eta_0,\eta_1)$& bandwidth $W$ & noise density $N_0$\\
  (9.61, 0.16, 1, 20) & (12.08, 0.11, 1.6, 23) & (4.88, 0.43, 0.1, 21) & 10 KHz & 5$\times10^{-15}$ W/Hz \\
  \hline
  \hline
  carrier frequency $f$ & BS coverage area $S_{\rm{bs}}$ & velocity $v_{\rm{h}}$, $v_{\rm{a}}$, $v_{\rm{d}}$ & data rate $C$  & considered area size\\
  2.4 GHz & $10^{4}$ $\rm{m}^2$ & 1 m/s & 10 Kbps & 1000 m $\times$ 1000 m\\
  \bottomrule
 \end{tabular}
\end{table}

\section{Numerical Results}\label{Sec:Simulations}
In this section, we present some numerical results to validate our theoretical founds. Besides, more insights on the effectiveness of proposed optimal placement strategy and the corresponding S-MGD based placement updating method are provided. The common parameters are listed in Table \ref{tab:ComParameters}\cite{bor2016efficient}. In our simulations, without specification, the simulated communication environment is Urban. The considered area is a square, where two rectangle-shaped subregions with equal area are investigated. A RSC is deployed at the center of considered area. Without loss of generality, let ${S_\beta}/{(\pi E_{\rm{b}})}=1~\rm{m}^2/\rm{J}$, $\beta=1,2$ and $P_{\rm{h}},P_{\rm{a}},P_{\rm{d}}=P_{\rm{m}}$, where $P_{\rm{m}}$ is the mobility power.

\subsection{Optimal Placement Strategy in Single Time-slot}

\begin{figure}
\centering{\begin{minipage}{1.0\linewidth}	
  \centerline{\includegraphics[width=8.0cm]{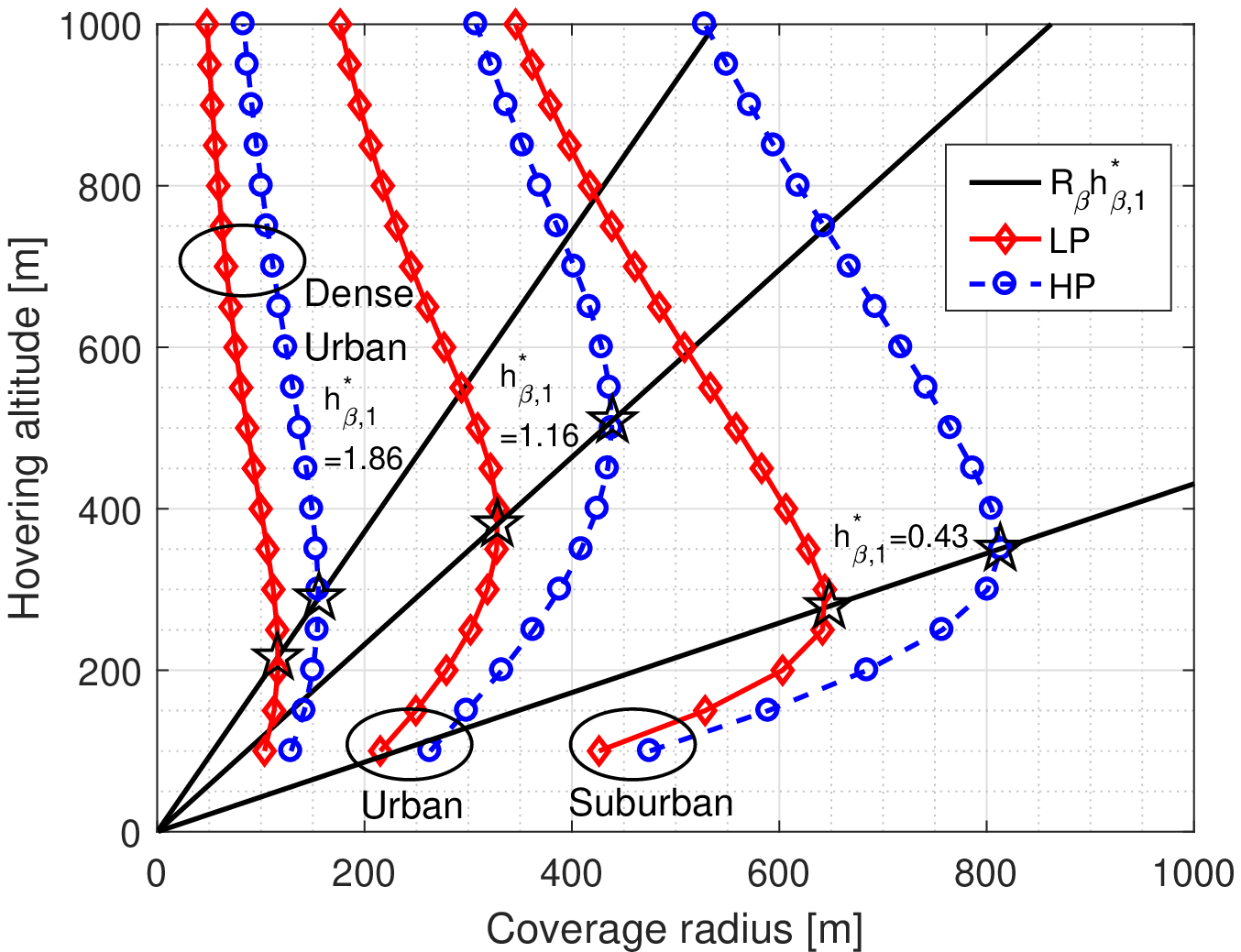}}
  \caption{The hovering altitudes versus $R_{\beta}$ in various environments. The optimal hovering altitude is marked by stars.} \label{Fig:HversusRWithFixedPower}
\end{minipage}}

\begin{minipage}{0.48\linewidth}
  \centerline{\includegraphics[width=8.0cm]{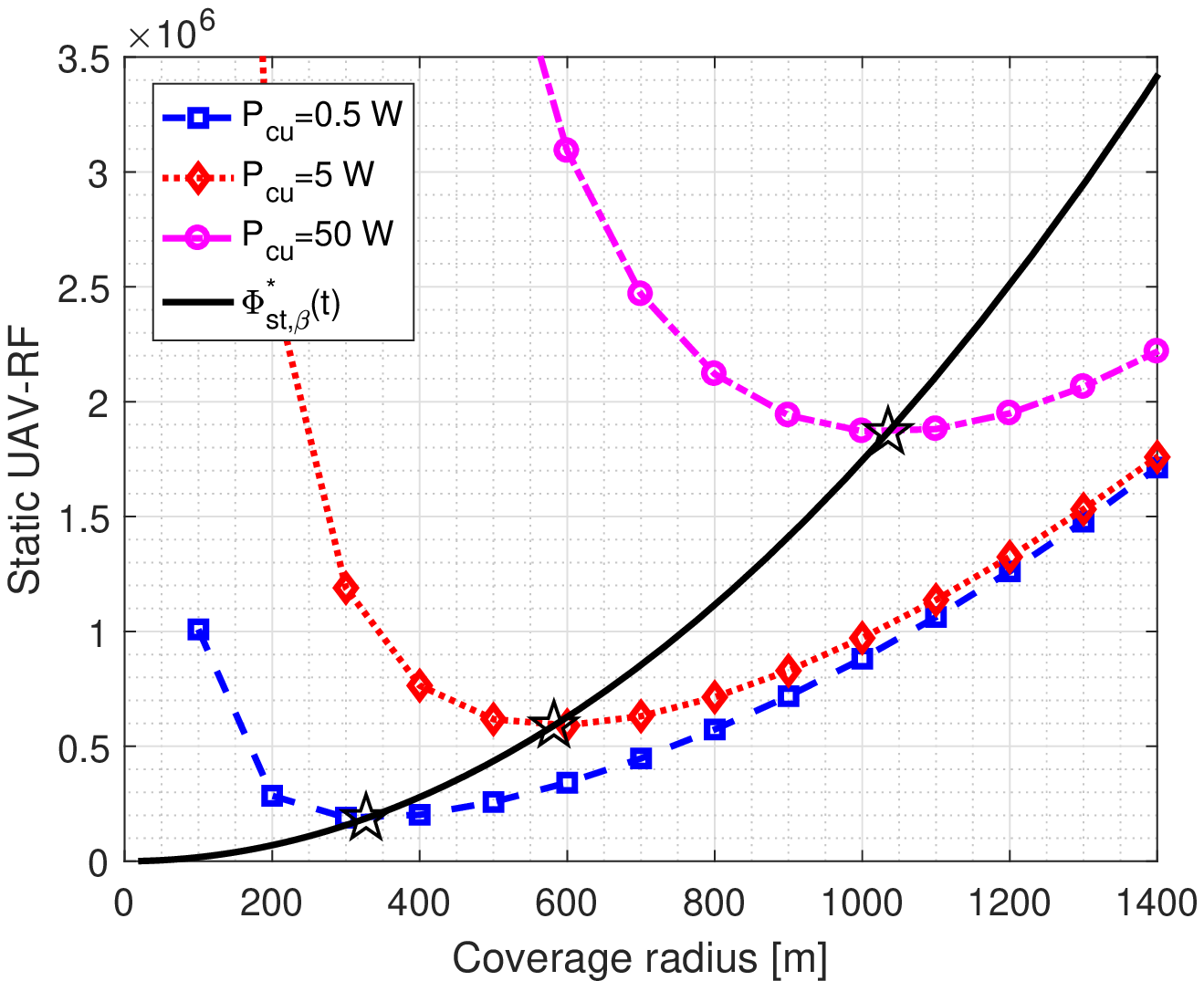}}
  \caption{The static UAV-RF versus coverage radius with various on-board circuit power. $\lambda_{\beta}$=0.1 /$\rm{m}^2$.} \label{Fig:PhiVersusR}
\end{minipage}
\hfill
\begin{minipage}{0.48\linewidth}
  \centerline{\includegraphics[width=8.0cm]{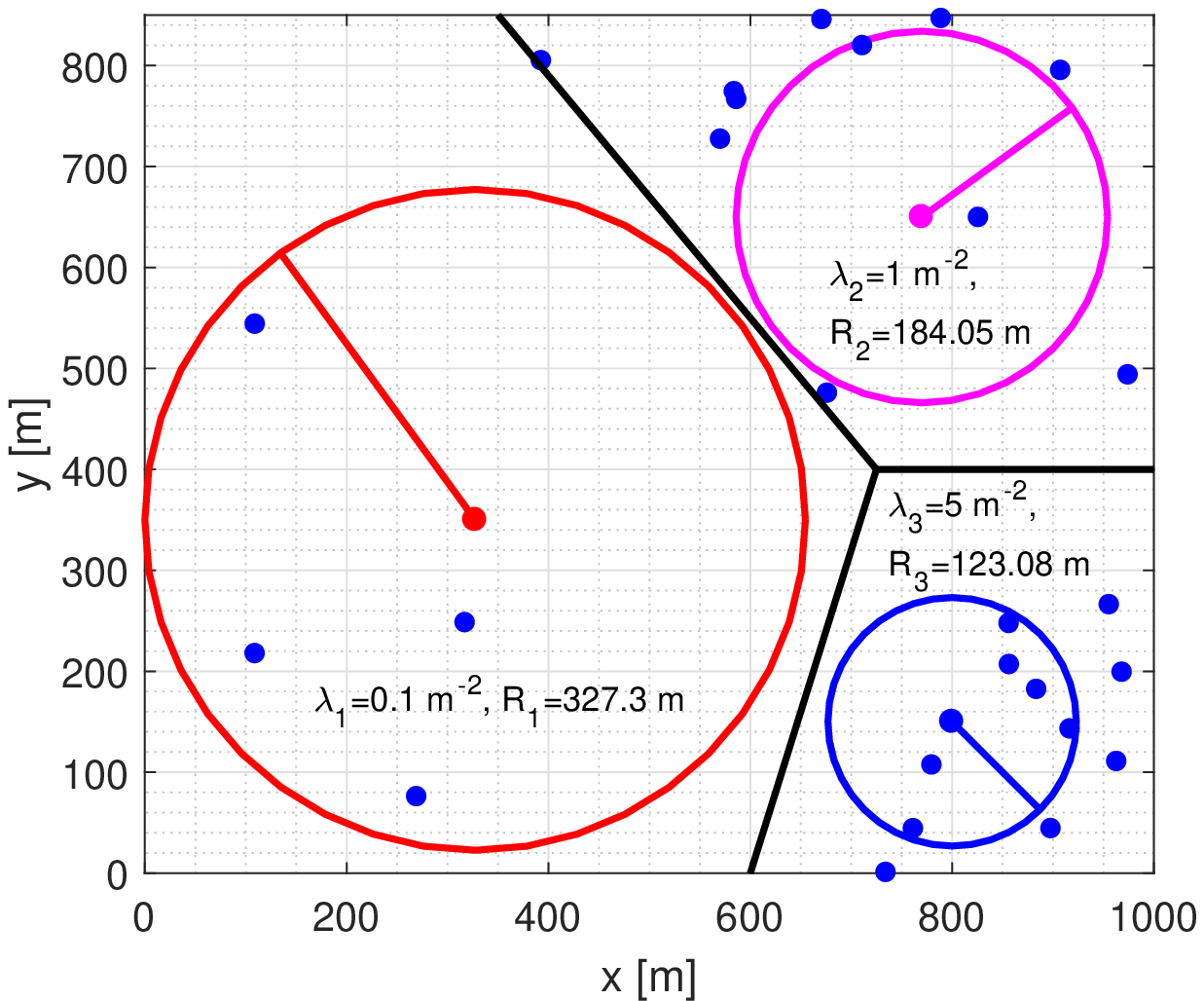}}
  \caption{An example of the coverage system under different densities of ground users. $P_{\rm{cu}}$=0.5 W.} \label{Fig:CoverageExample}
\end{minipage}
\end{figure}

Fig. \ref{Fig:HversusRWithFixedPower} depicts the hovering altitude versus coverage radius in various environments when the transmit power is fixed. In dense urban, urban and suburban environments, the red-solid lines correspond to transmit power 0.05 W, 0.5 W and 1.5 W. Similarly, the blue-dash lines correspond to 0.15 W, 1.5 W and 5 W. The simulated density of ground users is 0.1 /$\rm{m}^2$, and the optimal hovering altitude is marked by stars. Observing the blue-dash line in suburban, one can find that when the transmit power is fixed, the coverage radius achieves its maximum 810 m at $h_\beta=350$ m. In other words, when the coverage radius is 810 m, $h_\beta=350$ m is the optimal hovering altitude that minimizes the transmit power. The solid black lines depict the optimal hovering altitude with respect to $R_{\beta}$. It can be seen that the optimal hovering altitudes is proportional to the desired coverage radius, which verifies our theoretical results shown in Lemma \ref{Lem:ScaleOptimalAltitude}. In addition, it can be observed that in environment with high-rise buildings, the optimal hovering altitude is also high. This is because high hovering altitude can reduce shade effects from high-rise buildings, and hence the LOS probability is increased. In this way, the average path loss is decreased and transmit power is saved.

The static UAV-RF versus coverage radius with various on-board circuit power is depicted in Fig. \ref{Fig:PhiVersusR}. The density of ground users is 0.1 /$\rm{m}^2$, and the optimal coverage radii that minimize $\Phi_{\rm{st},\beta}(t)$ are marked by stars. When $P_{\rm{cu}}$=0.5 W, 5 W and 50 W, the simulated $R^*_{\beta}$=327.3 m, 582 m and 1035 m. As expected, high on-board circuit power leads to high static UAV-RF, which has been shown in  \eqref{Equ:InequalityOfPhi}. This indicates that restricting the on-board circuit power of UAV-BSs can effectively decrease the static UAV-RF. Also, it can be observed that the optimal coverage radius increases with respect to on-board circuit power. This is because when $P_{\rm{cu}}$ is high, large $R_{\beta}(t)$ can decrease the number of UAV-BSs, resulting in the reduction of the total consumed on-board circuit power of network.

Fig. \ref{Fig:CoverageExample} shows an example of coverage system under different densities of ground users, when $P_{\rm{cu}}$=0.5 W. The related user densities in three subregions are 0.1 $/\rm{m}^2$, 1 $/\rm{m}^2$ and 5 $/\rm{m}^2$, respectively. The corresponding theoretical coverage radii given by  \eqref{Equ:Optimal2DArrangement} are 327.3 m, 184.05 m and 123.08 m, which agree with the simulated results. To the hovering altitudes, they can be obtained with Lemma \ref{Lem:ScaleOptimalAltitude}. It can be observed that in cases with dense users, the optimal coverage is small. This is because small coverage radius can reduce the transmit power increased by high user density. Comparing the simulation results in considered subregions, we can conclude that our proposed optimal placement strategy can efficiently adjust to varying user densities, while minimizing the static UAV-RF.

\subsection{Optimal Placement Updating Method in Considered Duration}

To demonstrate the effectiveness of our proposed S-MGD method, we illustrate two robust placement updating methods for comparison. The first one is the 'Lazy' method, where UAV-BSs hold their placement in our considered duration. In contrast, The second one is the 'Diligent' method, where UAV-BSs update their positions at the beginning of each time-slot. In this way, the static UAV-RF of 'Diligent' method, denoted as $\Phi^{\rm{D}}_{\rm{st},\beta}(t)$, is always minimized. When take the energy cost on mobility into consideration, $\Phi^{\rm{D}}_{\rm{st},\beta}(t)$ can be denoted as the lower-bound of $\Phi_{\rm{dn},\beta}(t)$. Here, we take $\Phi^{\rm{D}}_{\rm{st},\beta}(t)$ as the base line to evaluate the effectiveness of our proposed method. In this subsection, $P_{\rm{cu}}=0.5$ W and $T$=24 h. 

Fig. \ref{Fig:UpdatePoints} depicts the updated coverage radii with our proposed S-MGD scheme, which reflects the placement of UAV-BSs. Clearly, the number of update epochs depends on the mobility power. Compare the cases with different mobility powers, one can observe that in cases with low mobility power, the coverage radii of UAV-BSs updates frequently. By contrast, in cases with high mobility power, UAV-BSs tend to keep their coverage radii for several time-slots. Specifically, when $P_{\rm{cu}}$=50 W, the coverage radii of UAV-BSs remain unchanged during $[0,T]$. This is because updating coverage radii may cost more energy than it saves.

\begin{figure*}[htbp]
\centering
\includegraphics[width=1\textwidth]{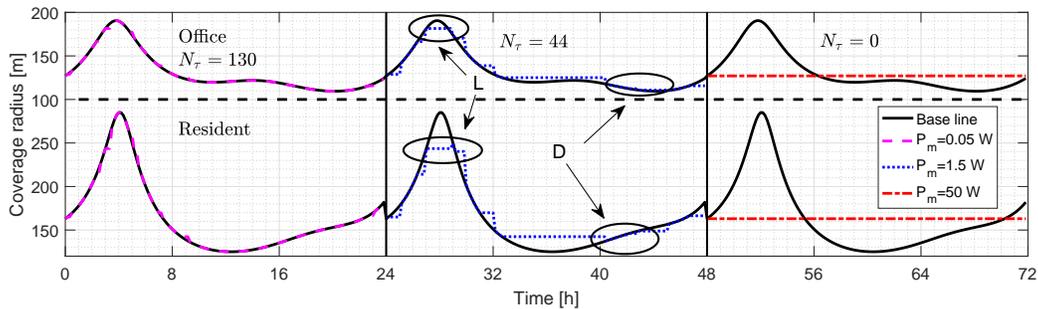}
\caption{The coverage radii of UAV-BSs in cases with different mobility powers. $P_{\rm{cu}}=0.5$ W and $T$=24 h.} \label{Fig:UpdatePoints}
\end{figure*}

In addition, Fig. \ref{Fig:UpdatePoints} shows that the placement updating frequency relies on the varying density of ground users. For example, in the case with $P_{\rm{m}}$=1.5 W, 'L' and 'D' labels the lazy and diligent part of our proposed S-MGD scheme. It can be seen that in lazy part, the placement of UAV-BSs remain unchanged. By contrast, in diligent part, the placement is updated in each time-slot. This is intuitive, because in lazy part, the optimal coverage radii in both considered subregions return to the initial values at the end of this part. Therefore, there is no need to update the placement. Also, in diligent part, the optimal coverage radii vary sharply. Thus, spending energy on updating the placement can save more energy cost on communication and on-board circuit. Simulation results in Fig. \ref{Fig:UpdatePoints} demonstrate that our proposed S-MGD based placement updating method can efficiently adjust to the varying user density and mobility power.

\begin{figure*}[htbp]
\centering
\includegraphics[width=1\textwidth]{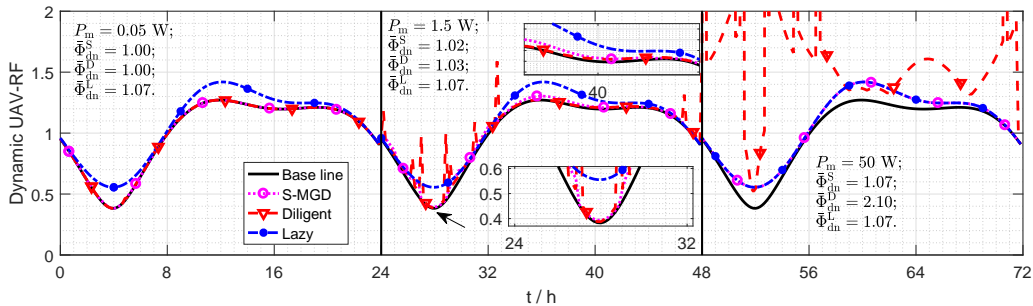}
\caption{The UAV-RF in cases with different mobility powers. $P_{\rm{cu}}=0.5$ W and $T$=24 h.} \label{Fig:CompareVariousRF}
\end{figure*}

Fig. \ref{Fig:CompareVariousRF} compares the dynamic UAV-RF of Diligent, Lazy and our proposed S-MGD based placement updating method. The simulated mobility power is 0.05 W, 1.5 W and 50 W. Fig. \ref{Fig:CompareVariousRF} shows that S-MGD always provides the most energy efficiency updating strategy for arbitrary mobility powers. Also, as expected, in cases with high and low mobility power, the proposed method degrades to 'Lazy' and 'Diligent' method, respectively. Observe the labeled average dynamic UAV-RF, we can see that compared with the worst updating method, our method can reduce the average dynamic UAV-RF by 7\%$\sim$96\%. 

In addition, since the 'Lazy' method holds the placement of UAV-BSs during $[0,T]$, the average dynamic UAV-RF remains unchanged for arbitrary mobility powers. Compared with the base line, the 'Lazy' updating method can provide satisfactory energy efficiency. However, Fig. \ref{Fig:AverageUAVRFVersusStartPoint} shows that the performance of 'Lazy' method is unstable and may vary with the initial time of our considered duration. By contrast, our method can stably provide near-optimal performance in arbitrary cases, regardless of the initial time of considered duration and the mobility powers.

\begin{figure}
\centering{\begin{minipage}{1.0\linewidth}	
  \centerline{\includegraphics[width=8.0cm]{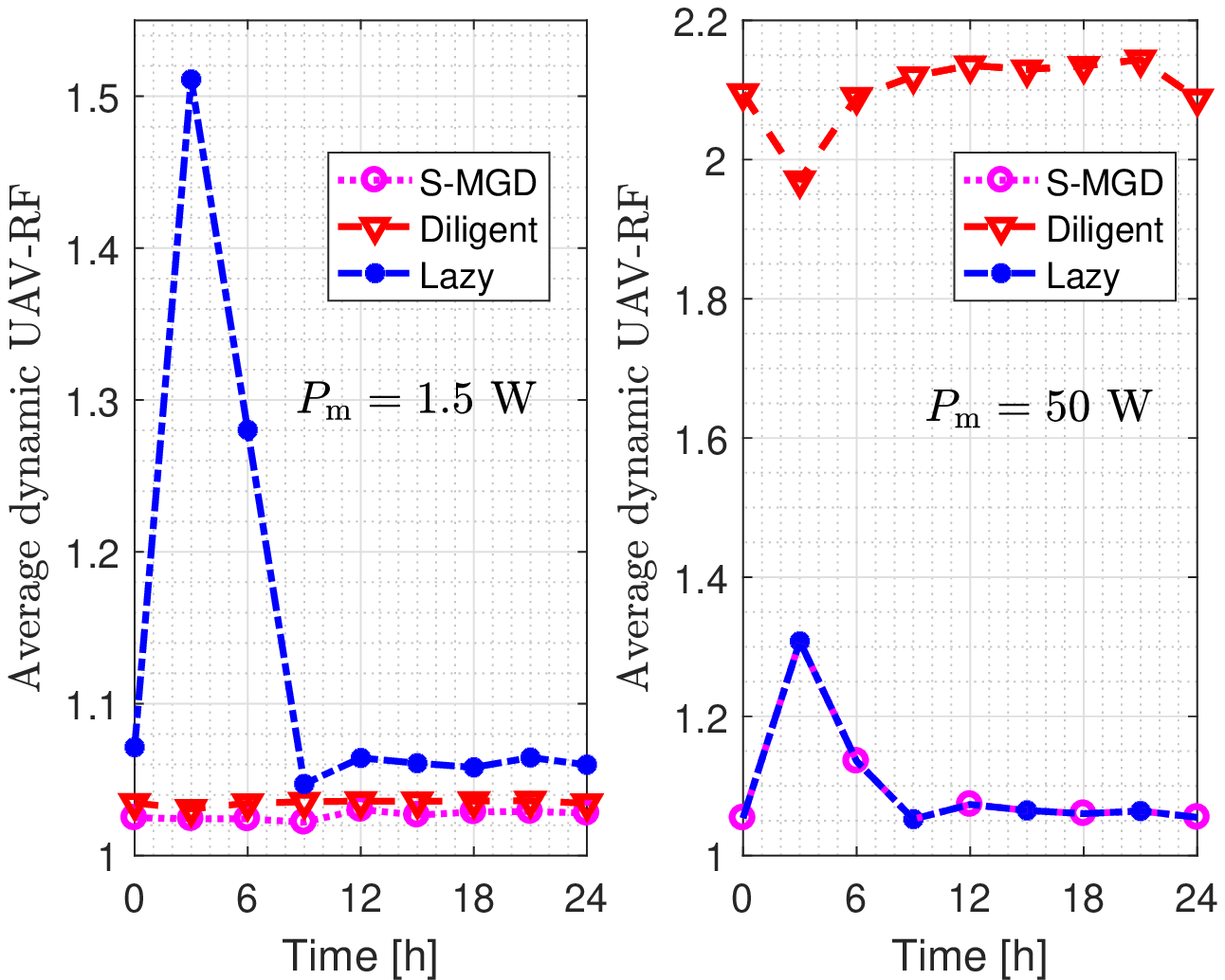}}
  \caption{The average dynamic UAV-RF of UAV-BSs versus different initial time.} \label{Fig:AverageUAVRFVersusStartPoint}
\end{minipage}}

\begin{minipage}{0.48\linewidth}
  \centerline{\includegraphics[width=8.0cm]{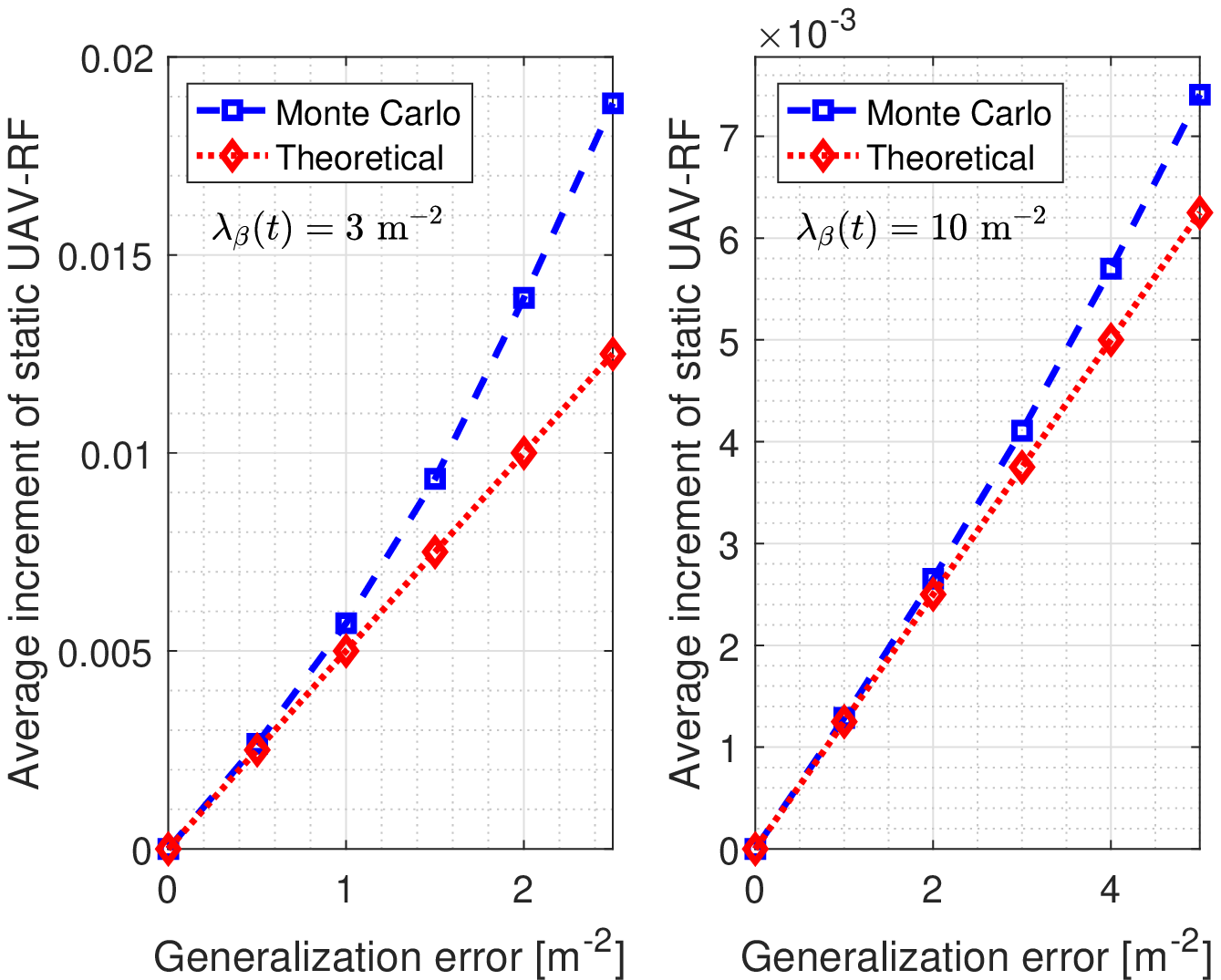}}
  \caption{The simulated and theoretical expected increase of static UAV-RF versus user density. Repetition number $10^6$.}\label{Fig:MonteIncrease}
\end{minipage}
\hfill
\begin{minipage}{0.48\linewidth}
  \centerline{\includegraphics[width=8.0cm]{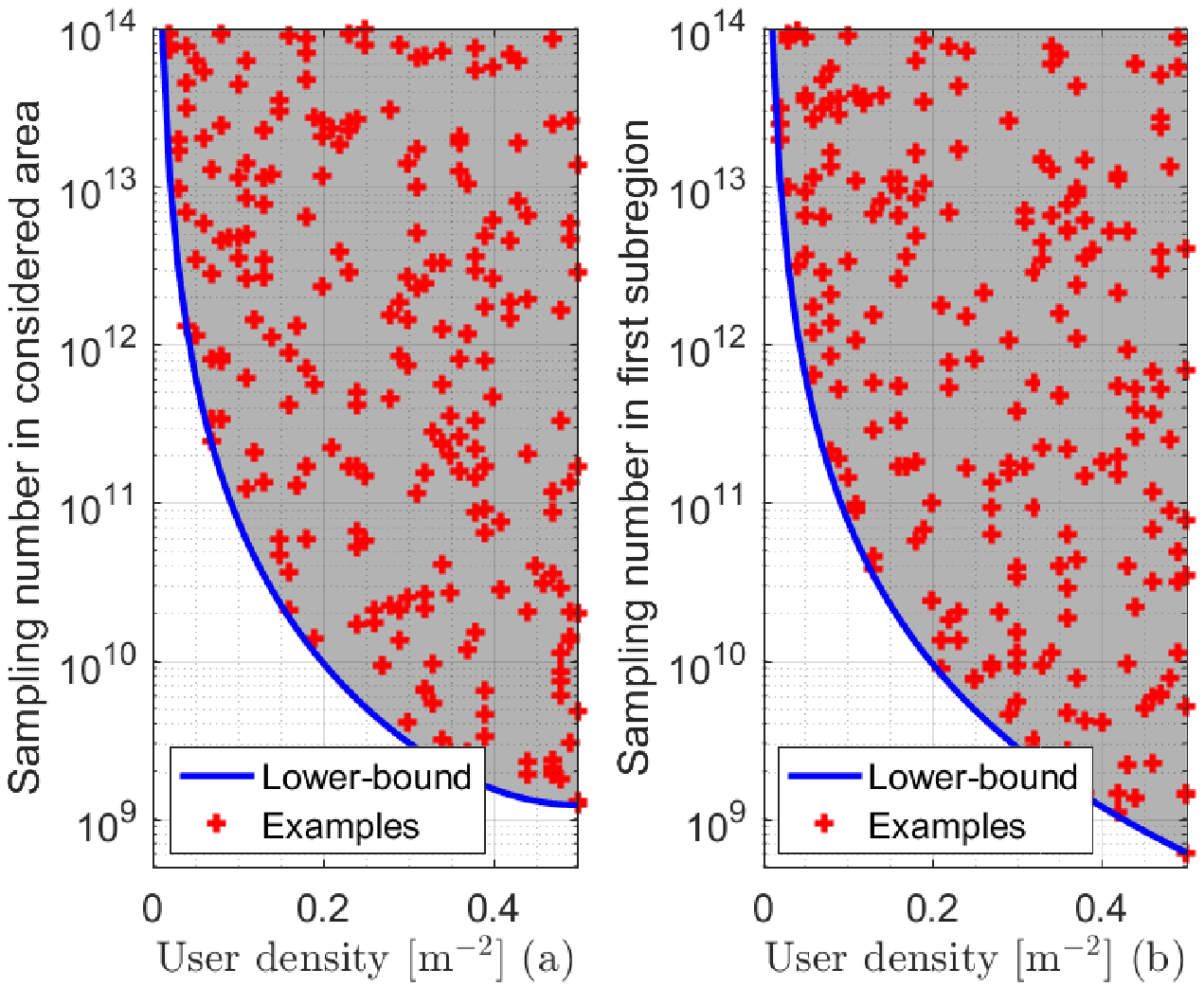}}
  \caption{The lower-bound of sampling numbers in considered area (a) and the first subregion (b) versus user density $\lambda_\beta$.} \label{Fig:MonteExamples}
\end{minipage}
\end{figure}

\subsection{Insights on Learning Density Patterns}

Fig. \ref{Fig:MonteIncrease} compares the simulated and theoretical $\Delta \Phi_{\rm{st},\beta}(t)$ with Monte-Carlo method, where the repetition number is $10^6$. For the simplicity of comparison, the depicted $\Delta \Phi_{\rm{st},\beta}(t)$ is normalized by $\Phi_{\rm{st},\beta}(t)$. Obviously, our theoretical results shown in Theorem \ref{thm:biasAndVarOnStatic} are sufficiently explicit in most of the simulated cases. However, the difference between our derived theoretical results and the simulated $\Delta \Phi_{\rm{st},\beta}(t)$ grows with respect to generalization error. This is because Proposition \ref{lem:minSampleNumber} is derived with Taylor series expansion, which lost accuracy for large $\hat{\lambda}_{\beta} - \lambda_{\beta}$. In addition, compare the results shown in cases with $\lambda_{\beta}=$3 $\rm{m}^{-2}$ and 10 $\rm{m}^{-2}$, we can see that our theoretical results in Theorem \ref{thm:biasAndVarOnStatic} are more accurate when user density is high. This is because as previous illustrated, $\Delta \Phi_{\rm{st},\beta}(t)$ in subregions with low user density are more sensitive to generalization error.

The lower-bound of sampling numbers in considered area and subregion versus user density $\lambda_{\beta}$ are depicted in Fig. \ref{Fig:MonteExamples}(a) and Fig. \ref{Fig:MonteExamples}(b), respectively. In simulations, $\Delta \Phi_{\rm{st},\max}(t)$=10, and without loss of generality, let $\xi_{\max}$=0. For our considered subregions, the simulated density of users is $\lambda_{\beta}$ and $1-\lambda_{\beta}$, where $\lambda_{\beta} \in [0,1/2]$. Also, some sampling numbers satisfying  \eqref{equ:SNOptimizeProblemConstraintsb} are labeled with Monte-Carlo method. It can be observed that the sampling numbers are lower-bounded by  \eqref{equ:SNOfBeta}, which verifies Proposition \ref{lem:minSampleNumber}. Also, Fig. \ref{Fig:MonteExamples}(b) shows that cases with low user density require large sample sets. This is consistent with the intuition that in cases with lower user densities, $\Delta \Phi_{\rm{st},\beta}(t)$ is more sensitive to generalization error. This manifests that cases with low user densities may effect the energy performance of network greatly and need to be taken into consideration carefully.

\section{Conclusion}\label{Sec:Conclusion}

This paper focused on the placement of UAV-BSs for the downlink system, where the powers relevant to signal transmitting, on-board circuit and the potential mobility of UAVs as well as the density of ground users are taken into account. To track the instable and non-ergodic time-varying nature of user density, we first proposed a pattern formation based framework in a machine learning manner. Then, to minimize the UAV-RF of network, the optimal placement of UAV-BSs was analyzed. In the static case with one time-slot, we proved that the optimal placement is achieved when the transmit power of UAV-BSs equals their on-board circuit power. In addition, in the dynamic case with a multiple time-slot duration, we showed that the optimal placement updating problem is an nonlinear dynamic programming coupled with an inherent integer linear programming. Considering the NP-hardness of original problem, we proposed a S-MGD based Pareto-optimal placement updating method with a polynomial time complexity. Simulation results showed that it can stably provide near-optimal performance in terms of minimum UAV-RF. Finally, we proved that the generalization errors of learned density patterns proportionally contribute to the increment of UAV-RF in each subregion. Simulation and theoretical results demonstrated that large sample sets are needed in regions with large area, high-rise buildings or low user density. 

Further, it is seen from Remark \ref{Rem:Complexity} that the computational complexity of proposed S-MGD method is greatly affected by the number of UAV-BSs. In previous mentioned regions, the computational complexity can be extremely high. Therefore, to reduce the complexity, multiple RSCs and proper partition to considered region are needed. In addition, the effect of pattern formation accuracy on the dynamic UAV-RF should be discussed. These issues shall be considered in future works.






\end{document}